\documentclass[journal]{IEEEtran}
\usepackage{hyperref}
\usepackage{dsfont}
\usepackage[utf8x]{inputenc} 
\usepackage{amsmath}
\usepackage{mathrsfs}
\usepackage{amsfonts}
\usepackage{amssymb,amsthm}
\usepackage{graphicx}
\usepackage{color}
\usepackage{todonotes}
\newcommand{\norm}[1]{\left\lVert#1\right\rVert}
\newcommand{\pdev}[2]{\frac{\partial #1}{\partial #2}}
\newcommand{\minlamb}[1]{\lambda_{\tt {min}}(#1)}
\newcommand{\maxlamb}[1]{\lambda_{\tt{max}}(#1)}
\newcommand{\inv}[1]{#1^{-1}}
\newcommand{\tinv}[1]{#1^{-\top}}
\def\rea{\mathds{R}}
\def\cmp{\mathds{C}}
\newtheorem{remark}{Remark}
\newtheorem{theorem}{Theorem}
\newtheorem{assumption}{Assumption}
\newtheorem{corollary}{Corollary}
\newtheorem{proposition}{Proposition}
\newtheorem{definition}{Definition}

\def\BibTeX{{\rm B\kern-.05em{\sc i\kern-.025em b}\kern-.08em
		T\kern-.1667em\lower.7ex\hbox{E}\kern-.125emX}}
\begin{document}
	\title{Tuning of passivity-based controllers for mechanical systems}

	\author{Carmen Chan-Zheng, Pablo Borja, and Jacquelien M.A. Scherpen
		\thanks{The work of C. Chan-Zheng is sponsored by the University of Costa Rica. }
		\thanks{C. Chan-Zheng and J.M.A. Scherpen are with the Jan C. Willems Center for Systems and Control, and Engineering and Technology institute Groningen (ENTEG), Faculty of Science and Engineering at the University of Groningen, Groningen 9747 AG, The Netherlands (email: c.chan.zheng@rug.nl, j.m.a.scherpen@rug.nl).}
		\thanks{P. Borja is with the School of Engineering, Computing and Mathematics, University of Plymouth, Plymouth, Devon PL4 8AA, United Kingdom. (email: pablo.borjarosales@plymouth.ac.uk). }}
	\maketitle
	\thispagestyle{empty}
	\begin{abstract}
	This manuscript describes several approaches for tuning the parameters of a class of passivity-based controllers for standard nonlinear mechanical systems. Particularly, we are interested in tuning controllers that preserve the mechanical system structure in the closed loop. To this end, first, we provide tuning rules for stabilization, i.e.,  the rate of convergence (exponential stability) and stability margin (input-to-state stability). Then, we provide guidelines to remove the overshoot. Additionally, we propose a methodology to tune the gyroscopic-related parameters. We also provide remarks on the damping phenomenon to facilitate the practical implementation of our approaches. 	We conclude this paper with experimental results obtained from applying our tuning rules to a fully-actuated and an underactuated mechanical system. 
	\end{abstract}

	\section{Introduction}	
	The modeling and control of mechanical systems have been widely studied and reported in the literature due to their fundamental role in industries such as aerospace, automotive, biomedics, semiconductors, or manufacturing.

	For modeling,  we find the port-Hamiltonian (pH) framework among the existing approaches. This framework is an energy-based modeling technique that represents a large class of nonlinear physical systems from different domains \cite{duindam2009modeling,van2014port}. Moreover, it highlights the physical properties of the system under study. In particular, for the mechanical domain,  this modeling approach underscores the role of the interconnection structure, dissipation, potential, and kinetic energy play in the system behavior. Furthermore, the passivity property of the system can be verified by selecting the total energy of the system~--~i.e., the Hamiltonian~--~as the storage function. 
	
	On the other hand, amidst the existing control strategies to stabilize pH mechanical systems, we study the passivity-based control (PBC) methodologies, a set of well-established techniques that offer a constructive approach for stabilizing a large class of complex systems \cite{ortega2013passivity,vanderSchaft2017}. In general, these techniques consist of two main steps: i) the so-called \textit{energy shaping} (resp. \textit{power shaping}) process, and ii) the \textit{damping injection}. The former step modifies the total energy (resp. power) of the system to guarantee that the closed-loop system has a {stable} equilibrium at the desired point; additionally, the interconnection structure of the closed-loop system may be modified as a result of this step.  Then, the second step ensures asymptotic stability properties for the desired equilibrium point. Some results of PBC approaches for stabilizing mechanical systems can be found in \cite{gomez2004physical,romero2018global,acosta2005interconnection,viola2007total,hamada2020passivity,wesselink}.
	
	Customarily, the control parameters of PBC approaches are selected such that its closed-loop system exhibits a prescribed performance \textit{in terms of stability}. For instance, in \cite{vanderSchaft2017,borja2020,Donaire2012,romero2014globally} we find results on $\mathcal{L}_2$ stability, asymptotic stability, input-to-state (ISS) stability, and exponential stability (ES). However, sometimes, it is not sufficient to only prescribe a performance in terms of stability for several real applications \cite{wen2004experimental,bechlioulis2010prescribed}. For instance, it is essential to ensure a prescribed performance in terms of other indices (e.g., oscillations, rate of convergence, among others) to solve a task from applications involving mechanical systems that require high precision. 
	
	{In contrast to the linear counterpart~--~where we can find a substantial amount of results on tuning the linear (PID) controller \cite{aastrom2004revisiting}~--~the literature on tuning the gains for any nonlinear controller (including PBC schemes) is relatively scarce as the characterization of the nonlinear phenomena is in several cases an open problem. Moreover, despite the well-known theoretical advantages (for example, stability guarantees and improved performance despite the nonlinearities phenomena) of the nonlinear schemes compared with their linear counterpart, there is an evident gap between the practitioners and the theorists. This breach stems from the fact that implementing the nonlinear schemes is challenging and rarely seen in real-life applications as there are no guidelines to achieve a desirable performance. Moreover, an additional challenge in creating tuning methodologies in the nonlinear field is that there is no unified framework for characterizing the frequency domain for the nonlinear systems. Although, there are few tuning methods for nonlinear approaches~--~e.g., Neural Networks \cite{chen1990back,rodriguez2021self}~--~stability guarantees remain a challenge.
		
	Amid the PBC schemes, we find the interconnection and damping assignment (IDA)-PBC methodology, which is a universally stabilizing controller in the sense that it generates all asymptotically stabilizing controllers for systems that can be represented in the pH structure \cite{ortega2002interconnection}. To tune this scheme, the author in \cite{kotyczka2013local} proposes a methodology that consists of prescribing local dynamics to the closed-loop system via the eigenvalue assignment approach. However, the gain selection process from this methodology lacks physical intuition. Additional results for other PBC approaches can be found in \cite{ferguson2019kinetic,chan2021exponential,hamada2020passivity,jeltsema2004tuning,dirksz2013tuning,chan2020tuning}, where they demonstrate that the parameters can be associated directly with the physical quantities of the closed-loop system~--~e.g.,  damping or energy. In \cite{ferguson2019kinetic,chan2021exponential}, the authors explore the relationship of the parameters with the decay ratio of the system trajectories via a particular choice of Lyapunov candidates; while in \cite{hamada2020passivity,jeltsema2004tuning,dirksz2013tuning,chan2020tuning} we find results for tuning the gains to remove the oscillations exhibited during the transient response. Moreover, in \cite{wesselink}, we find results on tuning the gyroscopic-related forces, where the authors demonstrate an improved performance in terms of the oscillations; however, no theoretical background is provided. The inclusion of the gyroscopic forces is critical for stabilizing some underactuated mechanical applications (see \cite{blankenstein2002matching,chang2002equivalence,woolsey2004controlled,borja2022role}).}
	
	Nonetheless, to the best of the authors’ knowledge, there is no comprehensive set of tuning methodologies in the literature for PBC approaches on standard mechanical systems that prescribe the performance of the closed-loop system in terms of other indices rather than the stability. In this manuscript, we provide tuning methodologies for PBC approaches that in closed-loop with a standard mechanical system result in some particular target dynamics that preserve the mechanical structure. Our main contributions are listed as follows.
	\begin{itemize}
		\item [(i)] {We extend the results in \cite{chan2021exponential,chan2020tuning}, where we underscore that these are obtained for a particular PBC scheme, namely, the PID-PBC (see \cite{borja2020}). The current manuscript generalizes the results by considering a larger class of PBC schemes whose closed loops recover the port-Hamiltonian and the standard mechanical structures (we find the IDA-PBC among these schemes). Moreover, here we consider the effect of the gyroscopic forces~--~~that may be introduced in other PBC approaches~--~on the closed-loop performance, whose insight is absent in the previous work.  }
		\item [(ii)] We provide a novel guideline to select the upper bound of the maximum overshoot permissible for the output of the closed-loop system.
		\item [(iii)] We present novel results based on the ISS property, where we provide an expression in terms of the PBC parameters for the stability margin of the closed-loop system. 
		\item [(iv)] {We discuss the role of damping and energy on the performance (in terms of different metrics) of the system by relating these physical quantities to the PBC parameters}.
		\item [(v)] {We present a novel insight into the effect of the PBC gyroscopic-related parameters on the performance of the closed-loop system.}
\end{itemize}

	The remainder of this paper is structured as follows: in Section \ref{prel} we provide the theoretical backgrounds and formulate the problem under study. Section \ref{nonlinear} describes the tuning rules derived from the ES and ISS analysis. Section \ref{trtuning} provides the tuning rules to prescribe the behavior of the closed-loop system in the vicinity of the desired equilibrium. In Section \ref{treatment}, we provide some remarks on the practical implementation of the tuning rules. Then, we describe the experimental results obtained from two separate configurations: i) a 5 DoF robotic arm (a fully-actuated mechanical system), and ii) a two degrees-of-freedom (DoF) planar manipulator with flexible joints (an underactuated mechanical system),  in Section \ref{case}. We conclude this manuscript with some remarks and future work in Section \ref{conclusion}.
	
	\textbf{Notation}: We denote the $n\times n$ identity matrix as $I_n$ and the $n\times m$ matrix of zeros as $0_{n\times m}$. For a given smooth function $f:\rea^n\to \rea$, we define the differential operator $\nabla_x f:=\frac{\partial f}{\partial x}$ which is a column vector, and $\nabla^2_x f:=\frac{\partial^2 f}{\partial x^2}$. For a smooth mapping $F:\rea^n\to\rea^m$, we define the $ij-$element of its ${n\times m}$ Jacobian matrix as $(\nabla_x F)_{ij}:=\frac{\partial F_i}{\partial x_j}$. When clear from the context the subindex in $\nabla$ is omitted. For a given vector $x\in\rea^{n}$, we say that $A$ is \textit{positive definite (semi-definite)}, denoted as $A\succ0$  ($A\succeq0$), if $A=A^{\top}$ and $x^{\top}Ax>0$ ($x^{\top}Ax\geq0$) for all $x\in \rea^{n}-\{0_{n} \}$ ($\rea^{n}$). For a given vector $x\in\rea^n$ , we denote the Euclidean norm as $\norm{x}$ and the $\mathcal{L}_2$-norm as $\norm{x}_2$. For a given matrix $B\in\rea^{n\times m}$ , we denote its largest singular value as $\sigma_{\max}(B)$. For $B=B^\top$, we denote by $\maxlamb{B}$ (resp. $\minlamb{B}$) as the maximum (resp. minimum) eigenvalue of $B$. 
	Given a distinguished element $x_\star\in\rea^n$, we define the constant matrix $B_\star:=B(x_\star)\in\rea^{n\times m}$. We denote $\rea_+$ as the set of strictly positive real numbers and $\rea_{\geq0}$ as the set $\rea_+\cup \{0\}$. Let $x,y\in\rea^{n}$, we define $col(x,y):=[x^\top y^\top]^\top$. We denote $e_i$ as the $i^{th}$ element of the canonical basis of $\rea^n$. All the functions considered in this manuscript are assumed to be (at least) twice continuously differentiable. 
	
	\textbf{Caveat}: when possible, we omit the arguments to simplify the notation.

	\section{Preliminaries and Problem formulation}\label{prel}
	In this section, we provide the pH representation of standard mechanical systems considered throughout this manuscript. Moreover, we present the target dynamics and provide a brief discussion on some PBC approaches that can achieve such dynamics. Then, we describe a particular pH structure, namely, \textit{the canonical Hamiltonian system}. Additionally, we discuss some stability properties and conclude this section with the problem formulation.
	
	\subsection{Description of standard mechanical systems}
	Consider a standard mechanical system in the pH framework
	\begin{equation}\label{sysmec}
		\arraycolsep=1pt \def\arraystretch{1.6}
		\begin{array}{rcl}
			\begin{bmatrix}
				\dot{{q}}\\\dot{{p}}
			\end{bmatrix} = \begin{bmatrix}
				0_{n\times n}&I_n\\
				-I_n &-{D}({q,p})\end{bmatrix}\begin{bmatrix}
				\nabla_q{H}({q},{p}) \\ \nabla_{p}{H}({q},{p})
			\end{bmatrix}&+&\begin{bmatrix}
				0_{n\times m} \\ G(q)
			\end{bmatrix}u
		\end{array}
	\end{equation} 
	\begin{equation*}
		\arraycolsep=1pt \def\arraystretch{1.6}
		\begin{array}{rcl}
			{H}({q},{p})&=&\displaystyle\frac{1}{2}{p}^\top \inv{M}({q}) {p}+{U}({q}),~y=G(q)^\top\inv{M}{{(q)p}}
		\end{array}
	\end{equation*}
	where ${q},{p} \in \rea^{n}$ are the generalized positions and momenta vectors, respectively; ${H}:\rea^n\times\rea^n\to \rea$ is the Hamiltonian of the system; the potential energy of the system is denoted with ${U}:\rea^n\to\rea$; ${M:\rea^n\to\rea^{n\times n}}$ corresponds to the mass inertia matrix, which is positive definite; ${{D}:\rea^n\times\rea^n\to \rea^{n\times n}}$ is positive semi-definite and represents the natural damping of the system; $u,y \in \rea^{m}$ are the control and passive output vectors, respectively; $m\leq n$; and $G(q):\rea^n\to \rea^{n\times m}$ is the input vector with $rank(G)=m$, which we define~--~to ease the presentation of the results~--~as 
	$$
	G:=\begin{bmatrix}
		0_{\ell\times m}\\I_m
	\end{bmatrix},
	$$
	with $\ell:=n-m$. 
	
	The set of assignable equilibria for \eqref{sysmec} is defined by
	\begin{equation*}
		\mathcal{E}:=\{{q},{p} \in \rea^n \ | \ {p}=0_n, \ G^\perp \nabla{U(q)}=0_{\ell}\},
	\end{equation*}
	where $G^{\perp} := \left[I_{\ell} \ 0_{\ell\times m}\right]$.
	
	Moreover, for all $q$, $M(q)$ is bounded, that is, \begin{equation}\label{boundM}
		\lambda_{\min}(M(q))I_n\leq{M(q)}\leq\lambda_{\max}(M(q))I_n.
	\end{equation}
	We refer the reader to \cite{ghorbel1993positive} for a complete characterization of robot manipulators with bounded inertia matrix.
	
	\subsection{The target dynamics}
	The stabilization of mechanical systems via PBC has been extensively studied, see for instance \cite{ortega2013passivity, gomez2004physical,romero2018global,acosta2005interconnection,viola2007total,hamada2020passivity,wesselink}.  Additionally, the energy shaping process is translated to find a Hamiltonian with an isolated minimum at $(q_\star,0_n) \in \mathcal{E}$, where $q_\star \in \rea^n$ is the desired configuration. Moreover, for some PBC approaches, shaping the kinetic energy results directly in modifying the interconnection structure (see for instance \cite{viola2007total,wesselink,zhang2017pid}). On the other hand, the damping injection process is performed by feeding back the passive output~--~customarily, it corresponds to the velocity~--~and ensures that the equilibrium is asymptotically stable. 

	Although the aforementioned references provide guidelines to guarantee stability, they lack tuning methodologies to ensure performance in terms of other indices. Therefore, in this manuscript, we focus on providing tuning rules for PBC methodologies that obtain the following target dynamics 
	\begin{equation}\label{tg1}
		\begin{bmatrix}
			\dot{{q}}\\\dot{{p}}
		\end{bmatrix}=\left({J}_{d}({q},{p})-R_d(q,p)\right)\nabla {H}_{d}({q},{p})
	\end{equation} 
	with
	\begin{equation}\label{tg2}
		\arraycolsep=1pt \def\arraystretch{1.6}
		\begin{array}{rcl}
			{J}_{d}({q},{p})&:=& \begin{bmatrix}
				0_{n\times n}&\inv{M}({q})M_d({q})\\ -M_d({q})\inv{M}({q})& {J_2(q,p)}
			\end{bmatrix}\\[0.5cm]	
			{R}_{d}({q},{p})&:=& \begin{bmatrix}
			0_{n\times n}&0_{n\times n}\\ 0_{n\times n}& {D}_{d}({q,p})
			\end{bmatrix}\\[0.5cm]	
			{H}_d{(q,p)}&:=& \displaystyle\frac{1}{2}{p}^\top \inv{M_d}({q}){p}+{U}_d{(q)},	
		\end{array}
	\end{equation} 
	where ${H}_d:\rea^n\times \rea^n \to\rea_+$ is the desired Hamiltonian; the desired inertia matrix $M_d:\rea^n\to \rea^{n\times n}$ is positive definite;  the desired potential energy ${U}_d:\rea^n\to\rea_+$ has a \textit{locally isolated minimum} at $q_{\star}$; the desired damping matrix ${{D}_{d}:\rea^n \times \rea^n \to \rea^{n\times n}}$ is positive semi-definite; and ${{J_2}:\rea^{n}\times \rea^{n} \to\rea^{n\times n}}$ is skew-symmetric.
	
	Then, we consider the following assumption throughout this manuscript to obtain the tuning guidelines.
	\\
		\begin{assumption}\label{ass1}
		Given \eqref{sysmec} and the desired equilibrium $(q_\star,0_n)\in\mathcal{E}$, there exists a control approach $u\in\rea^m$ such that the target dynamics takes the pH form \eqref{tg1}-\eqref{tg2}. Moreover, the desired Hamiltonian $H_d(q,p)$ has a local isolated minimum at $(q_\star,0_n)$, that is, the closed-loop system \eqref{tg1}-\eqref{tg2} is stable.
		\hfill$\square$
		\\
	\end{assumption}

	{In other words, we are interested in tuning PBC approaches such that the closed-loop system preserves the mechanical structure~--~in addition to the pH one~--~as in \eqref{tg1}-\eqref{tg2}. We emphasize that preserving the structure is crucial for developing our tuning rules, whose main benefit is endowing with physical intuition the process of gain selection. Some PBC methodologies encountered in the literature that verify Assumption \ref{ass1} are reported in \cite{gomez2004physical, wesselink,borja2020}.  The stabilization of mechanical systems via IDA-PBC is described in \cite{gomez2004physical}. While in \cite{wesselink,borja2020} report PID-PBC approaches that do not require the solution of partial differential equations.}

	\subsection{The canonical Hamiltonian system}\label{phcanonical}
	A change of coordinates is a well-known tool for converting a particular system into another structure that may provide better insight into a particular feature of the system under study (see \cite{chan2021exponential,venkatraman2010speed,fujimoto2003trajectory}). In the current section, we describe a particular transformation for \eqref{tg1}-\eqref{tg2} whose resulting structure~--~namely, \textit{the canonical Hamiltonian system}~--~may highlight the effects of the gyroscopic forces on the behavior of the closed-loop system. The transformed system is given by	
		\begin{equation}\label{syscanonical}
		\arraycolsep=1pt \def\arraystretch{1.6}
		\begin{array}{rl}
			\begin{bmatrix}
				\dot{{q}}\\\dot{{p}_c}
			\end{bmatrix} &= (J_c-R_c)\begin{bmatrix}
				\nabla_q{H_c}({q},{p}) \\ \nabla_{p_c}{H_c}({q},{p}_c)
			\end{bmatrix}, 
		\end{array}
	\end{equation} 
	with
	\begin{equation}\label{syscanonical2}
		\arraycolsep=5pt \def\arraystretch{1.6}
		\begin{array}{rl}
			&J_c:=\begin{bmatrix}
				0_{n\times n}&I_n\\
				-I_n &0_{n\times n}\end{bmatrix}, \quad R_c:=\begin{bmatrix}
				0_{n\times n}&0_{n\times n}\\
				0_{n\times n} &D_c(q,p_c)\end{bmatrix} 
		\end{array}
	\end{equation} 
	where $D_c:\rea^n \times \rea^n\to \rea^{n\times n}$ is positive semi-definite, $H_c:\rea^n\times \rea^n \to \rea_+$ is the \textit{canonical Hamiltonian}, and $(q,p_c)$ are the \textit{canonical coordinates} of \eqref{tg1}-\eqref{tg2}.
	
	The process to transform \eqref{tg1}-\eqref{tg2} into the canonical Hamiltonian system is reported in \cite{blankenstein2002matching} where the authors describe a particular choice for $J_2(q,p)$. We summarize such process in the following proposition.
	
	\begin{proposition}\label{prop2}
		Let $$p_c:=M(q)\inv{M}_d(q)p+Q_d(q)$$ with  $Q_d:\rea^n\to\rea^n$ being any smooth vector-valued function. Then, the closed-loop system \eqref{tg1}-\eqref{tg2} results in the canonical Hamiltonian system \eqref{syscanonical}-\eqref{syscanonical2} if and only if
		\begin{equation}\label{prop2eq}
				J_2(q,p):=\hat{J}(q,p)+J_g(q)
		\end{equation}
		where 
		\begin{equation*}
		\arraycolsep=1pt \def\arraystretch{1.6}
		\begin{array}{rl}
		\hat{J}(q,p):=&M_d(q)\inv{M}(q)[(\nabla_q M(q)\inv{M}_d(q)p)^\top\\
			&-\nabla_q(M(q)\inv{M}_d(q)p)]\inv{M}(q)M_d(q),\\
		J_g(q):=	&M_d(q)\inv{M}(q)[(\nabla_{q}Q_d(q))^\top\\
			&-\nabla_{q}Q_d(q)]\inv{M}(q)M_d(q).
		\end{array}
	\end{equation*}
		Moreover, the Hamiltonian of \eqref{syscanonical}-\eqref{syscanonical2} takes the form
		\begin{equation}\label{canham}
		\arraycolsep=1pt \def\arraystretch{1.6}
			\begin{array}{rl}
				H_c(q,p_c)&:=\frac{1}{2}(p_c-Q_d(q))^\top \inv{M}_c(q)(p_c-Q_d(q))+U_d(q)\\
			\end{array}
		\end{equation}
		where
		\begin{equation*}
		\arraycolsep=1pt \def\arraystretch{1.6}
			\begin{array}{rl}
				M_c(q)&:=M(q)\inv{M}_d(q) M(q)\\
				D_c(q,p)&:=M(q)\inv{M}_d(q)D_d(q,p)\inv{M}_d(q)M(q).
			\end{array}
		\end{equation*}
	
	\hfill$\square$
	\end{proposition} 
	Note that $J_g(q)\neq0_{n\times n}$~--~equivalently, ${\nabla_q Q_d(q)\neq 0_{n\times n}}$ or $\nabla_q Q_d(q)\neq \nabla_q Q_d(q)^\top$~--~introduces gyroscopic-related forces into the closed-loop system \eqref{syscanonical}-\eqref{syscanonical2} via the term $p_c^\top \inv{M}_c(q)Q_d(q)$ from the canonical Hamiltonian \eqref{canham}.  The introduction of gyroscopic-related forces has interesting positive effects on the performance in terms of stabilization and oscillations.  For instance, in \cite{wesselink,chan2021passivity}, the authors demonstrate~--~via experiments~--~that the inclusion of $J_g(q)$ reduces the oscillations in some coordinates for underactuated mechanical systems; moreover, the addition of this term is crucial for stabilizing underactuated mechanical applications such as spacecraft control and underwater vehicle control \cite{blankenstein2002matching,chang2002equivalence,borja2022role}.
	
	Furthermore, we can characterize the gyroscopic terms by dividing it into \textit{intrinsic} or \textit{non-instrinsic} of which we provide the definition as follows
	\begin{definition}[Intrinsic gyroscopic terms \cite{blankenstein2002matching}]\label{def1}
		The gyroscopic terms are called \textit{intrinsic} if there \textit{does not} exist a canonical transformation $(q,p_c)\mapsto (\tilde{q}_c,\tilde{p}_c)$ such that the Hamiltonian in the new coordinates takes the form of the kinetic plus the potential energy, i.e.,
		\begin{equation*}
			\tilde{H}_c(\tilde{q}_c,\tilde{p}_c) = \frac{1}{2}\tilde{p}_c ^\top \inv{\tilde{M}}(\tilde{q}_c)\tilde{p}_c+\tilde{U}(\tilde{q}_c),
		\end{equation*}
	for some $\tilde{M}:\rea^n\to\rea^{n\times n}$ and $\tilde{U}:\rea^n\to\rea_+$.
	\\
	\end{definition}
	
	Then, the following proposition verifies the intrinsic property.
	\begin{proposition}\label{prop3}
		The gyroscopic terms are intrinsic to the closed-loop system \eqref{syscanonical}-\eqref{syscanonical2} if and only if $J_g\neq 0_{n\times n}$.
	\hfill$\square$
	\end{proposition}
	For further details on Proposition \ref{prop3}, see \cite{blankenstein2002matching}.
	\subsection{Some stability properties}
	Throughout this paper, we consider two stability properties for the system \eqref{tg1}-\eqref{tg2}: the ES and ISS. The former ensures that an exponential decay function bounds the closed-loop system trajectories, while the latter ensures that the system trajectories are bounded for any initial conditions as long as the input is also bounded. Although both properties reveal interesting behaviors of the closed-loop system, these are only \textit{qualitative} attributes. Therefore, they may not provide explicit information for tuning purposes. In order to use these properties in a \textit{quantitative} manner, we provide some concepts that allow us to exploit the stability properties for tuning purposes.
	
	To this end, let us first introduce a definition related to the ES property.
	\begin{definition}[Rate of convergence \cite{khalil2002nonlinear}]
		The {rate of convergence} of the closed-loop system \eqref{tg1}-\eqref{tg2} is the exponential decay value of the trajectories of the system approaching the equilibrium  $(q_\star,0_n)$. We can characterize this value by defining some constants $k_1,k_2,k_3\in\rea_+$ such that
		\begin{equation*}
			\norm{col(q,p)}\leq \sqrt{\frac{k_2}{k1}}\norm{col(q_0,p_0)}\exp\left\{-\frac{k_3}{2k_2}(t-t_0)\right\},
		\end{equation*}
		where $\frac{k_3}{2k_2}$ corresponds to the upper bound of the \textit{rate of convergence}, $t_0\geq0$ is the initial time, and $q_0,p_0\in\rea^n$ are the initial conditions.
		\\
	\end{definition}
	
	On the other hand, consider the closed-loop system \eqref{tg1}-\eqref{tg2} with a disturbance signal, that is
	\begin{equation}\label{disturbed}
		\begin{bmatrix}
			\dot{{q}}\\\dot{{p}}
		\end{bmatrix}=({J}_{d}({q},{p})-{R}_{d}({q},{p}))\nabla {H}_{d}({q},{p})+d(t,q,p)
	\end{equation} 
	where $d:\rea_{\geq0}\times\rea^n\times\rea^n\to\rea^{2n}$ is a vector of disturbances, satisfying $\norm{d(t,q,p)}\leq \infty$. Then, we provide the following ISS related definitions:	
	\begin{definition}[Comparison functions \cite{khalil2002nonlinear}]
		~
		\begin{itemize}
			\item A continuous function $\alpha:[0,a)\mapsto[0,\infty)$ is said to belong to class $\mathcal{K}$ if it is strictly increasing and $\alpha(0)=0$. Moreover, if $\alpha\in\mathcal{K}$, $a=\infty$ and $\alpha(r)\to \infty$ as $r\to\infty$, then it belongs to class $\mathcal{K}_\infty$.
			\item A continuos function  $\alpha:[0,a)\times [0,\infty)\mapsto[0,\infty)$ is said to belong to class $\mathcal{KL}$ if i) for each fixed $s$, the mapping $\alpha(r,s)$ belongs to class $\mathcal{K}$ with respect to r, and ii) for each fixed $r$, the mapping $\alpha(r,s)$ is decreasing with respect to s and $\beta(r,s)\to0$ as $s\to\infty$.
			\\
		\end{itemize}
	\end{definition}
	\begin{definition}[Stability margin \cite{sontag1995characterizations}]\label{smdef}
		Consider system \eqref{disturbed}. Then, the nonlinear stability margin is any function $\rho\in\mathcal{K}_\infty$ that verifies
		\begin{equation}\label{smargin}
			\norm{d(t,q,p)}\leq\rho(\norm{col(q,p)})
		\end{equation}
		and 
		\begin{equation}\label{smargin2}
			\norm{col(q,p)}\leq\beta(\norm{col(q_0,p_0)},t) ~\forall t\geq0,
		\end{equation}
		where $\beta\in \mathcal{K}\mathcal{L}$.
		
		Moreover, the system \eqref{disturbed} is said to be ISS if and only if \eqref{smargin} and \eqref{smargin2} are satisfied.
			\\
	\end{definition}

	\subsection{Problem formulation}	
	Given $(q_\star,0_n)\in\mathcal{E}$, propose tuning methodologies to select the system matrices $M_d(q,p)$, $J_2(q,p)$, $D_d(q,p)$, and Hamiltonian $H_d(q,p)$ such that the closed-loop system \eqref{tg1}-\eqref{tg2} exhibits a prescribed behavior.
	
	These values are referred as  \textit{control parameters} for the rest of this manuscript.

	\section{Quantifying the ES and ISS properties}\label{nonlinear}
	
	In this section, we exploit the ES and ISS properties in a \textit{quantitative} manner by selecting an appropriate Lyapunov function candidate, of which we deduce the tuning rules for the upper bound of the \textit{rate of convergence}, \textit{maximum permissible overshoot}, and the \textit{stability margins} of the system. Towards this end, we split the section into two parts: i) first, we describe rules for the non-perturbed system \eqref{tg1}-\eqref{tg2}, and then, ii) we provide guidelines for the perturbed system \eqref{disturbed}.
	
	\subsection{Tuning guidelines for the non-perturbed system}\label{nonperturbed}
	
	Interesting properties for \eqref{tg1}-\eqref{tg2} are revealed through a convenient choice of a Lyapunov candidate. In \cite{chan2021exponential}, the authors provide a tuning rule for the upper bound of the rate of convergence obtained from Lyapunov stability analysis. 
   	In our manuscript, we extend such an approach, where we exploit the Lyapunov candidate to deduce the upper bound of the rate of convergence for a broader class of mechanical systems stabilizable via a larger class of PBC techniques. We also provide a novel expression for tuning the upper bound of the maximum permissible overshoot. To this end, consider the following assumption
   	\\
	\begin{assumption}\label{ass2}
		The control parameters from the closed-loop system \eqref{tg1}-\eqref{tg2} are chosen such that
		\begin{itemize}
			\item [\textbf{C1}.] $U_d(q)$ is strongly convex.
			\item[\textbf{C2}.] $\norm{M_d(q)}<c$ for some positive constant satisfying ${c<\infty}$.
			\item[\textbf{C3}.] ${D}_{d}(q,p)\succ0$.
		\end{itemize}
	\end{assumption}
	~\\
	Then, we prove  ES stability of the equilibrium point for the closed-loop system \eqref{tg1}-\eqref{tg2} in the following theorem
	\begin{theorem}\label{prop4}
		Consider the closed-loop system \eqref{tg1}-\eqref{tg2}, the desired equilibrium $x_\star:=(q_\star,0_n)\in\mathcal{E}$, and Assumption \ref{ass2}. Then,
		\begin{itemize}
			\item [(i)] $x_\star$ is an exponentially stable equilibrium.
			\item [(ii)] $x_\star$ is globally exponentially stable if $U_d(q)$ is radially unbounded.
	\end{itemize}
	\end{theorem}
	
	\begin{proof}
		To prove (i), consider the matrix decomposition
		\begin{equation*}
			\inv{M_d}(q)= T_d(q) T_d^\top(q), 
		\end{equation*}
		where $T_d:\rea^n\to\rea^n\times\rea^n$ is a full rank upper triangular matrix with strictly positive diagonal entries (see Cholesky decomposition \cite{horn2012matrix}). Furthermore, we introduce the change of coordinates (described first in \cite{venkatraman2010speed})
		\begin{equation}\label{change}
			\hat{q}:=q-q_\star,~\hat{p}:=T_d(q)^\top p.
		\end{equation}
		
		Then, define $$\hat{x}:=col(\hat{q},\hat{p}),$$ 
		
		and by transforming \eqref{tg1}-\eqref{tg2} with \eqref{change}, we get the new pH system
		
		\begin{equation}\label{transformed}
			\def\arraystretch{1.6}
			\begin{array}{rcl}
				\begin{bmatrix}
					\dot{\hat{q}}\\\dot{\hat{p}}
				\end{bmatrix}&=&\begin{bmatrix}
					0_{n\times n} &\hat{A}(\hat{q})\\
					-\hat{A}^\top(\hat{q})&\hat{J}(\hat{x})-\hat{D}(\hat{x})
				\end{bmatrix}\begin{bmatrix}
					\nabla_{\hat{q}}\hat{H}(\hat{x})\\\nabla_{\hat{p}}\hat{H}(\hat{x})\end{bmatrix}	\\
				\hat{H}(\hat{x})&=&\frac{1}{2}\hat{p}^\top\hat{p}+\hat{U}(\hat{q}),
			\end{array}
		\end{equation}
		where
		{\begin{equation*}
				\def\arraystretch{1.6}
			\begin{array}{rcl}
				\hat{U}(\hat{q})&:=&U_d(\hat{q}+q_\star)\\
				\hat{T}_d(\hat{q})&:=&T_d(\hat{q}+q_\star)	\\
				\hat{A}(\hat{q})&:=&\inv{M}(\hat{q}+q_\star)\hat{T}_d^{-\top}(\hat{q})\\	\hat{D}(\hat{x})&:=&\hat{T}_d^\top(\hat{q})D_d(\hat{q}+q_\star,\tinv{T_d}\hat{p})\hat{T}_d(\hat{q})\\
				\hat{J}(\hat{x})&:=&\hat{T}_d^\top(\hat{q}) J_2(\hat{q}+q_\star,\tinv{T_d}\hat{p})\hat{T}_d(\hat{q})+\\
				&&\displaystyle \sum^n_{i=1}\Big\{\left[\hat{p}^\top\inv{\hat{T}_d}(\hat{q})\pdev{\hat{T}_d(\hat{q})}{q_i}\right]^\top \left[\hat{A}^\top(\hat{q}) e_i\right]^\top-\\
				&& \left[\hat{A}^\top(\hat{q}) e_i\right]\left[\hat{p}^\top\inv{\hat{T}_d}(\hat{q})\pdev{\hat{T}_d(\hat{q})}{q_i}\right]\Big\}	
			\end{array}.
		\end{equation*}}
		
		Note that by implementing the change of coordinates \eqref{change}, it follows that the new equilibrium $\hat{x}_\star$ is the origin. Now, we prove the ES of the origin for \eqref{transformed} by considering the Lyapunov candidate
		\begin{equation}\label{candidate}
			S(\hat{x}):=\hat{H}(\hat{x})+\epsilon\hat{p}^\top\hat{A}^\top(\hat{q})\nabla_{\hat{q}}\hat{U}(\hat{q}),
		\end{equation}
		with $\epsilon\in\rea_+$.

		Then, due to Assumption \ref{ass2}, note that $\hat{H}(\hat{x})$ satisfies the bounds
		\begin{equation}\label{Hub}
			\frac{\beta_{\min}}{2}\norm{\hat{x}}^2\leq\hat{H}(\hat{x})\leq\frac{\beta_{\max}}{2}\norm{\hat{x}}^2,
		\end{equation}
		with
		\begin{equation}\label{betas}
			\begin{array}{rcl}
				\beta_{\max}&:=&\max\{1,\maxlamb{\nabla^2_{\hat{q}} \hat{U}(\hat{q})}\}\\
				\beta_{\min}&:=&\min\{1,\minlamb{\nabla^2_{\hat{q}} \hat{U}(\hat{q})}\}.
			\end{array}
		\end{equation}
		Furthermore, by applying Young's inequality\footnote{For $a,b\in\rea$ and $\epsilon_y \in\rea_+$, Young's inequality is given by $ab\leq\frac{a^2}{2\epsilon_y}+\frac{\epsilon_y b^2}{2} $}, we have that 	
		\begin{equation}\label{C}
			\begin{split}
				\norm{\epsilon{\hat{p}}^\top \hat{A}^\top\nabla_{\hat{q}}\hat{U}}
				\leq\frac{\epsilon\sigma_{\max}(\hat{A})\beta_{\max}^2}{2} \norm{\hat{x}}^2.
			\end{split}
		\end{equation}
		Also, note that from Assumption \ref{ass2}, we get the following chain of implications
		\begin{equation*}
			\begin{array}{rc}
			\norm{M_d}<\infty\implies\norm{\hat{T}_d}<\infty
			\implies \norm{\hat{A}}<\infty.		
			\end{array}
		\end{equation*}
		Thus, from \eqref{Hub} and \eqref{C}, we get 
		\begin{equation}\label{cond1}
			\begin{split}
				&k_1\norm{\hat{x}}^2\leq S(\hat{x})\leq k_2\norm{\hat{x}}^2
			\end{split}
		\end{equation}
		with
		\begin{equation}\label{k1k2}
			\begin{array}{rc}
				k_1&:=\frac{\beta_{min}-\epsilon\sigma_{\max}(\hat{A})\beta_{\max}^2}{2},~
				k_2:=\frac{\beta_{\max}+\epsilon\sigma_{\max}(\hat{A})\beta_{\max}^2}{2}.
			\end{array}
		\end{equation}
		{Note that there exists a \textit{sufficiently small}  $\epsilon$ such that $k_1\in\rea_{+}$. Hence, $S(\hat{x})\in\rea_+$ for all $\hat{x}\neq 0_n$.}
		
		Then, via some computations, it follows that the derivative of $S(\hat{x})$ is given by  
		\begin{equation*}
			\begin{split}
				\dot{S}(\hat{x})&= -\nabla^\top \hat{H}(\hat{x}) \Upsilon_{\tt{sym}}(\hat{x})\nabla \hat{H}(\hat{x})
			\end{split}
		\end{equation*}
		where the matrix $\Upsilon_{\tt{sym}}(\hat{x})$ is defined as
		\begin{equation}\label{upsilonsym2}
		\arraycolsep=1pt \def\arraystretch{1.6}
			\begin{array}{lll}
				\Upsilon_{\tt{sym}}(\hat{x})&:= \begin{bmatrix}
					\Upsilon_{11}&\Upsilon_{12}\\
					\Upsilon_{12}^\top&\Upsilon_{22}
				\end{bmatrix},\\
				\Upsilon_{11}(\hat{q})&:=\epsilon (\hat{A}(\hat{q})\hat{A}^\top(\hat{q})+\hat{A}^\top(\hat{q})\hat{A}(\hat{q})),\\
				\Upsilon_{12}(\hat{x})&:=\frac{\epsilon}{2}[\hat{A}(\hat{q})[\hat{D}(\hat{x})-\hat{J}(\hat{x})]-\dot{\hat{A}}(q)],\\
				\Upsilon_{22}(\hat{x})&:=\hat{D}(\hat{x})-\epsilon(\hat{A}^\top(\hat{q}) \nabla_{\hat{q}}^2\hat{U}(\hat{q})\hat{A}(\hat{q})).
			\end{array}
		\end{equation} 
	
		Then, note that $\Upsilon_{\tt{sym}}$ must be positive definite for $S(\hat{x})$ qualifying as a suitable Lyapunov candidate. It follows that we can demonstrate that $\Upsilon_{\tt{sym}}\succ0$ by applying the Schur complement analysis, i.e., observe that $\Upsilon_{11}\succ0$ and there always exists a \textit{sufficiently small} $\epsilon$ such that $\Upsilon_{22}\succ 0$ and the Schur complement of $\Upsilon_{11}$ is also positive definite, that is,
		\begin{equation}\label{schur}
			\Upsilon_{11}-\Upsilon_{12}\inv{\Upsilon_{22}}\Upsilon_{12}^\top\succ0,
		\end{equation}
		since $\hat{D}(\hat{x})\succ 0$.
		
		Subsequently, let $\mu\in\rea_+$ be the minimum eigenvalue of $\Upsilon_{\tt{sym}}(\hat{x})$, then, it follows that
		\begin{equation}\label{cond2}
			\begin{array}{rc}
				\dot{S}(\hat{x})\leq& -\mu \norm{\nabla \hat{H}(\hat{x})}^2\leq-\mu\beta_{\max}^2\norm{\hat{x}}^2.
			\end{array}
		\end{equation}
		
		Therefore, from \eqref{cond1} and \eqref{cond2}, $\hat{x}_\star$ is an exponentially stable equilibrium point for \eqref{transformed} (see Theorem 4.10 from \cite{khalil2002nonlinear}). 
		
		To prove (ii), we have the following chain of implications:
		\begin{equation*}
			\hat{q}\to\infty,~\hat{p}\to\infty\implies\hat{U}(\hat{q})\to\infty\implies S(\hat{x})\to\infty.
		\end{equation*}
	\end{proof}
	\begin{remark}
	 {The term $\epsilon$ is used in i) ensuring that the Schur complement of $\Upsilon_{11}$ is positive definite, and ii) verifying that $S(\hat{x})\in\rea_{+}$ for all $\hat{x}\neq 0_n$, i.e., there always exists a \textit{sufficiently small} $\epsilon$ such that $k_1$ from \eqref{k1k2} is positive.}
		\\
	\end{remark}
\begin{remark}
		We remark that the proof from Theorem \ref{prop4} considers the closed-loop system \eqref{tg1}-\eqref{tg2}, which can be obtained through different PBC approaches. Thus, it is more general than \cite{chan2021exponential}, where a standard mechanical system is in closed-loop with a particular PBC approach (namely, PID-PBC).
	\end{remark}
	~\\
	Note that from \eqref{cond1}-\eqref{cond2} we get that
	\begin{equation*}
		\dot{S}(\hat{x})\leq-\frac{2\mu\beta_{\max}}{1+\epsilon\sigma_{\max}(\hat{A}(\hat{q}))\beta_{\max}}S(\hat{x}).
	\end{equation*}
	Then, via the comparison lemma (see \cite{khalil2002nonlinear}), we have that the solution of \eqref{transformed} is bounded, i.e., 
	\begin{equation}\label{boundX}
		\norm{\hat{x}}\leq\sqrt{\frac{k_2}{k_1}}\norm{\hat{x}_0}\exp\left\{-	\frac{\mu\beta_{\max}}{1+\epsilon\sigma_{\max}(\hat{A}(\hat{q}))\beta_{\max}}t\right\},
	\end{equation}
	where $\hat{x}_0\in\rea^{2n}$ are the initial conditions in the new coordinates.
	
	It follows that we can exploit the inequality \eqref{boundX} and obtain two expressions: i) an upper bound for the rate of convergence, and ii) an upper bound for the maximum overshoot of output of the system.
	
	For the latter, note that the transformed output ~--~i.e., ${\hat{y}:=G^\top\hat{T}(\hat{q})\hat{p}}$~--~verifies the following
	\begin{equation*}
		\begin{array}{rcl}
			\norm{\hat{y}}&=&\norm{G^\top\hat{T}_d(\hat{q})\hat{p}}\\
			 &\leq&\sigma_{\max}(G^\top\hat{T}_d(\hat{q}))\norm{\hat{p}}\\
			 &\leq&\sigma_{\max}(G^\top\hat{T}_d(\hat{q}))\norm{\hat{x}}.
		\end{array}
	\end{equation*}
	Then, from \eqref{boundX}, it follows that 
	$$ \norm{\hat{y}}\leq\xi\exp\left\{-	\frac{\mu\beta_{\max}}{1+\epsilon\sigma_{\max}(\hat{A}(\hat{q}))\beta_{\max}}t\right\},$$
	with	
	\begin{equation}\label{ov}
		\xi:=\sigma_{\max}(G^\top\hat{T}_d(\hat{q}))\sqrt{\frac{k_2}{k_1}}\norm{\hat{x}_0}
	\end{equation}
	where $k_1,k_2$ are defined in \eqref{k1k2}. Therefore, we have proven the following result.
	\begin{corollary}\label{cor1}
	 The convergence rates of the trajectories of \eqref{transformed} are upper bounded by
		\begin{equation}\label{rate}
			\frac{\mu\beta_{\max}}{1+\epsilon\sigma_{\max}(\hat{A}(\hat{q}))\beta_{\max}}.
		\end{equation} 
		Moreover, the maximum overshoot of the system output $\hat{y}$ is upper bounded by \eqref{ov}.
		\hfill $\blacksquare$
		\\
	\end{corollary}
	
	Note that \eqref{ov} and \eqref{rate} are expressed in terms of $\beta_{\max}$,$\beta_{\min}$, $\sigma_{\max}(\hat{A}(\hat{q}))$, $\epsilon$ and $\mu$. The parameters $\beta_{\max}$ or $\beta_{\min}$ can be computed easily from the potential energy, and $\sigma_{\max}(\hat{A}(\hat{q}))$ can be obtained directly from the kinetic energy. Conversely, the computation of $\epsilon$ and $\mu$ is a challenge; nonetheless, we can still employ other well-known tools to study the behavior of these parameters. For $\mu$, we can employ the Gershgorin circle theorem (see \cite{horn2012matrix} for further details), which defines circles containing the location of the spectrum of $\Upsilon_{\tt{sym}}$. Each circle is defined by the $i^{th}$ row elements (with $i=1\hdots n$), with the center being the diagonal element and the radius being the sum of the absolute values of the non-diagonal entries. For example, in Fig. \ref{gershgorin} we show the Gershgorin circles for some $\Upsilon_{\tt{sym}}$ (with $n=4$); note that by augmenting the diagonal elements of \eqref{upsilonsym2}, the centers of the circles are shifted to the right; consequently, $\mu$ may increment. 
	\begin{figure}[t]
	\centering
	\includegraphics[width=\columnwidth]{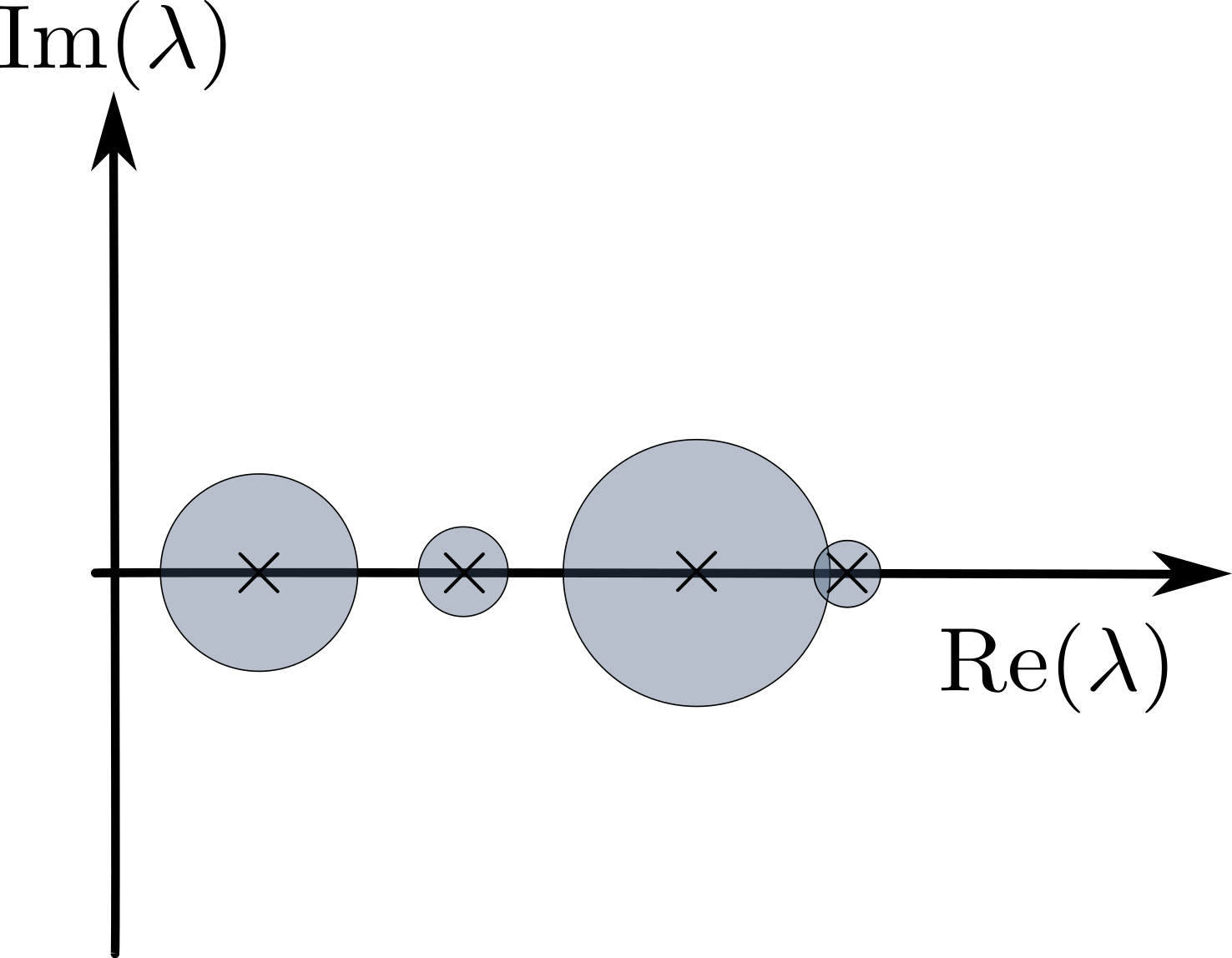}
	\caption{Gershgorin circles for $\Upsilon_{\tt{sym}}$}\label{gershgorin}
	\end{figure}
	As for $\epsilon$, we can employ the Schur complement analysis tool. For instance, note that  $\Upsilon_{12}$ from \eqref{upsilonsym2} increases as $\hat{D}(\hat{x})$ increases. Thus, it follows that {$\epsilon$ must be adjusted} to guarantee that condition \eqref{schur} still holds. 
	
	Since $\beta_{\max}$, $\beta_{\min}$, $\sigma_{\max}(\hat{A}(\hat{q}))$, $\epsilon$ and $\mu$ are related with the control parameters of the closed-loop systems, we can employ \eqref{ov} and \eqref{rate} as guidelines to prescribe the maximum permissible overshoot of the closed-loop system and the desired upper bound for the decay ratio of the trajectories, respectively.

	\begin{remark}
{	Note that \eqref{ov} and \eqref{rate} provide an upper bound of the overshoot and the rate of convergence, respectively. Thus, we can adjust the worst-case scenarios for these two performance indices by reducing the mentioned bounds. However, reducing the bounds does not necessarily affect the behavior of the system, especially if the bounds are conservative. Nevertheless, in the proof of Theorem  \eqref{prop4} and from the discussion above, note that the parameters from these indices are intimately related to the physical quantities of the system.}
	\end{remark}

	\subsection{Tuning guidelines for perturbed systems}\label{perturbed}
	The Lyapunov candidate \eqref{candidate} is conveniently chosen so that it reveals the effect of the control parameters on the rate of convergence and the maximum permissible overshoot of the closed-loop system. In this section, we further exploit this candidate selection by studying the effect of such parameters on the stability margin of the closed-loop system (see Definition \ref{smdef}). 	
	Hence, we consider the closed-loop system \eqref{disturbed} with the change of coordinates \eqref{change}, that is
	\begin{equation}\label{transformed2}
	\arraycolsep=1pt \def\arraystretch{1.6}
		\begin{array}{rcl}
			&\begin{bmatrix}
				\dot{\hat{q}}\\\dot{\hat{p}}
			\end{bmatrix}&=\begin{bmatrix}
				0_{n\times n} &\hat{A}(\hat{q})\\
				-\hat{A}^\top(\hat{q})&\hat{J}(\hat{x})-\hat{D}(\hat{x})
			\end{bmatrix}\begin{bmatrix}
				\nabla_{\hat{q}}\hat{H}(\hat{x})\\\nabla_{\hat{p}}\hat{H}(\hat{x})\end{bmatrix}+\hat{d}(t,\hat{q},\hat{p})\\[0.5cm]
		\end{array}
	\end{equation}
	$$\hat{H}(\hat{x})=\frac{1}{2}\hat{p}^\top\hat{p}+\hat{U}(\hat{q}),$$
	where $\hat{d}:\rea_{\geq0}\times\rea^n\times\rea^n\to\rea^{2n}$ is the time-dependent disturbance vector, satisfying $\norm{\hat{d}(t,\hat{q},\hat{p})}\leq\infty$. 
	Then, we state the following result.
	\begin{theorem}\label{prop5}
		The system \eqref{transformed2} is ISS with nonlinear stability margin
		\begin{equation}\label{sm}
			\rho(\norm{\hat{x}}):=g_r\norm{\hat{x}},
		\end{equation}
		where 
		\begin{equation}\label{gainmargin}
			g_r:=\frac{\mu\beta_{\max}\theta}{\varphi}
		\end{equation}
		is the gain margin of the closed-loop system \eqref{disturbed} with
		$0<\theta<1$, and some $\varphi, \mu, \beta_{\max} \in\rea_+$. 
	\end{theorem}
	\begin{proof}
		Consider the Lyapunov candidate \eqref{candidate}, the bounds \eqref{cond1}, and
		\begin{equation*}
			\hat{d}(t,\hat{q},\hat{p}):=\begin{bmatrix}
				\hat{d}_1(t,\hat{q},\hat{p})\\\hat{d}_2(t,\hat{q},\hat{p})
			\end{bmatrix}
		\end{equation*}
		where $\hat{d}_1,\hat{d}_2:\rea_{\geq0}\times\rea^n\times\rea^n\to\rea^{n}$.
		Then, via some computations, it follows that 
		\begin{equation*}
			\begin{split}
				\dot{S}=& -\nabla^\top \hat{H} \Upsilon_{\tt{sym}}\nabla \hat{H}+\\
				&\epsilon(\hat{d}_2^\top\hat{A}^\top\nabla_{\hat{q}}\hat{U}+\hat{d}_1^\top \nabla_{\hat{q}}^2\hat{U} \hat{A}\hat{p})+\hat{d}_1^\top\nabla_{\hat{q}} \hat{U}+\hat{d}_2^\top\hat{p},\\
				=&-\nabla^\top \hat{H} \Upsilon_{\tt{sym}}\nabla \hat{H}+\hat{d}^\top \Upsilon_d \nabla\hat{H}
			\end{split}
		\end{equation*}
	with
	\begin{equation*}
		\Upsilon_d(\hat{q}):=\begin{bmatrix}
			I_n&\epsilon\nabla^2_{\hat{q}}\hat{U}(\hat{q})\hat{A}(\hat{q})\\
			\epsilon\hat{A}^\top(\hat{q})&I_n
		\end{bmatrix}
	\end{equation*}
	and the matrix $\Upsilon_{\tt{sym}}(\hat{x})$ is defined as in \eqref{upsilonsym2}.
		
	Then, by applying the Schur complement analysis, it follows that always exists a \textit{sufficiently small} $\epsilon\in\rea_+$ such that \eqref{upsilonsym2} is positive definite. Let $\mu\in\rea_+$ be the minimum eigenvalue of $\Upsilon_{\tt{sym}}(q,p)$ and consider the bounds \eqref{Hub}.  Then, it follows that
		\begin{equation*}
			\begin{array}{rl}
				\dot{S}\leq& -\mu \norm{\nabla \hat{H}}^2+\sigma_{\max}(\Upsilon_d)\norm{\nabla \hat{H}}\norm{\hat{d}}\\
				\leq&-\mu \beta_{\max}^2\norm{\hat{x}}^2+\beta_{\max}\sigma_{\max}(\Upsilon_d)\norm{\hat{x}}\norm{\hat{d}}.
			\end{array}
		\end{equation*}

		Consider $0<\theta<1$; then, by rewriting the previous expression, we get
		\begin{equation*}
			\begin{array}{rcl}
				\dot{S}&\leq& -\mu \beta_{\max}^2\norm{\hat{x}}^2(1-\theta)+\beta_{\max}\sigma_{\max}(\Upsilon_d)\norm{\hat{x}}\norm{\hat{d}}\\
				&&-\mu \beta_{\max}^2\theta\norm{\hat{x}}^2.
			\end{array}
		\end{equation*}
		Therefore,  
		\begin{equation}\label{sdot4}
			\begin{array}{r}
				\dot{S}\leq -\mu \beta_{\max}^2\norm{\hat{x}}^2(1-\theta),~\forall \norm{\hat{x}}\in\Omega,	\end{array}
		\end{equation}
		where 
		\begin{equation}\label{omegaset}
			\Omega:=\left\{\hat{x}\in\rea^{2n}~\Big|~\norm{\hat{d}}\leq \rho(\norm{\hat{x}})\right\};
		\end{equation} 
		 and $\rho(\norm{\hat{x}})$ is defined as in \eqref{sm} with $\varphi:=\sigma_{\max}(\Upsilon_d)$.
		
	 	From \eqref{cond1} and \eqref{sdot4},  the closed-loop system \eqref{transformed2} is ISS with nonlinear stability margin $\rho(\norm{\hat{x}})$ (see Theorem 4.19 from \cite{khalil2002nonlinear} and \cite{sontag1995characterizations}).
	\end{proof}
~\\
	\begin{remark}
		Note that via the nonlinear stability margin concept, we can define the gain margin \eqref{gainmargin} that corresponds to the maximum permissible growth of the norm of the disturbance with respect to the norm of the trajectories in which the closed-loop system remains ISS. In other words, we get a better disturbance attenuation by increasing the gain margin. 
	\end{remark}

	Similar to Section \ref{nonperturbed} we can provide a decay ratio bound and a maximum overshoot for the system output. Note that from \eqref{cond1} and \eqref{sdot4}, we get
	\begin{equation*}
		\dot{S}(\hat{x})\leq-\frac{2\mu\beta_{\max}(1-\theta)}{1+\epsilon\sigma_{\max}(\hat{A})\beta_{\max}}S(\hat{x}).
	\end{equation*}
	Then, via the comparison lemma (see \cite{khalil2002nonlinear}), we have that the solution of \eqref{transformed2} is bounded, i.e.,
	\begin{equation*}
		\norm{\hat{x}}\leq\sqrt{\frac{k_2}{k_1}}\norm{\hat{x}_0}\exp\left\{-\frac{\mu\beta_{\max}(1-\theta)}{1+\epsilon\sigma_{\max}(\hat{A})\beta_{\max}}t\right\},
	\end{equation*}
	on $t\in(t_0,T]$ for some $T>0$. 
	
	Thus, we have proven the following result
	\begin{corollary}\label{cor2}
		Let 
		$$\Omega_e:={\{\hat{x}\in\rea^{2n}~|~\norm{\hat{x}}=\frac{1}{g_r}\norm{\hat{d}}\}}.$$ 
		Then, for some initial conditions $\norm{\hat{x}_0}\in \Omega$ (see \eqref{omegaset}), the trajectories approach exponentially to $\Omega_e$ at a rate of convergence that is upper bounded by
		\begin{equation}\label{rate2}
		\frac{\mu\beta_{\max}(1-\theta)}{1+\epsilon\sigma_{\max}(\hat{A})\beta_{\max}}
		\end{equation}
		as $t\to T$ for some $T>0$. 
		
		Furthermore, the maximum overshoot of the output of the system on $(t_0,T]$ is given by \eqref{ov}.
		\hfill $\blacksquare$
		\\
	\end{corollary}	
	By using simultaneously the expressions \eqref{gainmargin}, \eqref{rate2}, and \eqref{ov};  we provide an insight into the relationship of the control parameters with three performance metrics~--~i.e., the gain margin, the upper bound of the rate of convergence, and the maximum overshoot, respectively~--~of the perturbed system. \eqref{transformed2}. For instance, a trade-off between these metrics is evident by increasing $\beta_{\max}$, the disturbance $\hat{d}$ is attenuated, and the stability margin increases. Simultaneously,  the maximum overshoot is augmented.

	\begin{remark}
		Another approach for tuning the stability margin of general nonlinear systems can be found in \cite{grune2002input}, where the authors introduce an equivalent concept to ISS, namely, \textit{input-to-state dynamical stability (ISDS)}. However, this tuning methodology lacks physical intuition as there is no clear relation between the parameters associated with the stability margin with energy or damping.
	\end{remark}
		\begin{remark}
		By using the expression \eqref{sdot4}, we can calculate the $\mathcal{L}_2$-norm for the signal $\hat{x}$, i.e.,
		\begin{equation}\label{l2x}
			\arraycolsep=1pt \def\arraystretch{1.6}
			\begin{array}{rcl}
				\norm{\hat{x}}^2_2&=&\int^\infty_0\norm{\hat{x}}^2 d\tau\\
				&\leq&-\frac{1}{\mu\beta_{\max}^2(1-\theta)} \int^\infty_0\dot{S}(\tau) d\tau\\
				&\leq& \frac{1}{\mu\beta_{\max}^2(1-\theta)}(S(0)-S(\infty)) \\
				&\leq& \frac{1}{\mu\beta_{\max}^2(1-\theta)}S(0).
			\end{array}
		\end{equation}

		Additionally, consider \eqref{l2x} and since $\norm{\hat{d}}_2\leq g_r\norm{\hat{x}}_2$, we have that
		\begin{equation}\label{l2d2}
			\arraycolsep=1pt \def\arraystretch{1.6}
			\begin{array}{rcl}
				\norm{\hat{d}}_2^2&\leq& g_r^2 \int^\infty_0\norm{\hat{x}}^2 d\tau=g_r^2\norm{\hat{x}}^2_2\\
				&\leq&\frac{\mu\theta^2}{\sigma_{\max}(\Upsilon_d)^2(1-\theta)} S(0).
			\end{array}
		\end{equation}
		Recall that the square of the $\mathcal{L}_2$-norm of a signal corresponds to the energy contained in such signal. Therefore, \eqref{l2d2} provides the upper bound of the energy of the disturbance  in which the system remains stable. We remark the effect of the control parameters on such bound. Moreover, \eqref{l2d2} can be rewritten as
		$$\dfrac{\norm{\hat{d}}_2^2}{\norm{\hat{x}}_2^2}\leq{g_r^2};$$ 
		thus, we can see clearly from the previous expression that the maximum permissible growth of the energy of disturbance with respect to the energy of the trajectories~~--~~in which the system remains stable~~--~~is given by the square of the gain margin $g_r$.
	\end{remark}
	~\\
	\section{On the behavior of the closed-loop system near the equilibrium}\label{trtuning}
	In this section, we propose tuning rules to prescribe a desired performance in the transient response. To this end, we recur to the linearization of the closed-loop system \eqref{tg1}-\eqref{tg2}; and later, we find a transformation such that the linearized system has a \textit{saddle point matrix} structure (we refer the reader to \cite{benzi2005numerical,benzi2006eigenvalues} for further details of this class of matrices). This particular structure reveals interesting spectral properties for the linearized matrix, of which we deduce our tuning guidelines.  
	
	Let us first introduce the vectors $\tilde{q}:=q-q_\star$ and $\tilde{p}:=p-p_\star$ with $p_\star:=0_n$.  Then, it follows that the linearized system around the equilibrium point $(q_\star,0_n)$ corresponds to
	\begin{equation*}
		\begin{array}{rc}
			\begin{bmatrix}
				\dot{\tilde{q}}\\\dot{\tilde{p}}
			\end{bmatrix}=(J_{d\star}-R_{d\star})\nabla^2 H_{d\star}	\begin{bmatrix}
				\tilde{q}\\\tilde{p}
			\end{bmatrix}
		\end{array}.
	\end{equation*}
	
	Subsequently, consider the following Cholesky decompositions (see \cite{horn2012matrix})
	\begin{equation*}
		\inv{M}_{d\star}=\phi_M^\top\phi_M, \nabla_q^2U_{d\star}=\phi_P^\top\phi_P
	\end{equation*}
	where $\phi_M,\phi_P\in\rea^{n\times n}$ are upper triangular  matrices; then, consider the similarity transformation matrix $T\in\rea^{2n\times 2n}$ and new coordinates $z\in\rea^{2n}$
	\begin{equation*}
		T:=\begin{bmatrix}
			0_{n\times n}&\phi_M\\\phi_P&0_{n\times n}
		\end{bmatrix},~z=T\begin{bmatrix}
			\tilde{q}\\\tilde{p}
		\end{bmatrix}.
	\end{equation*}
	Thus, the linearized system in the newly introduced coordinates $z$ becomes
	\begin{equation}\label{linearized}
		\begin{array}{rll}
			\dot{z}&=&-\mathcal{A}z,\\
			\mathcal{A}&:=&\begin{bmatrix}
				\phi_M(D_{d\star}-J_{2\star})\phi_M^\top&\tinv{\phi_M}\inv{M}_\star\phi_P^\top\\
				-\phi_P\inv{M}_\star\phi_\inv{M}&0_{n\times n}
			\end{bmatrix}.
		\end{array}
	\end{equation}
	
	Then, inspired by the results of Brayton and Moser \cite{brayton1964theory}, 
 	we provide a proposition on the location of the spectrum of $\mathcal{A}$ from \eqref{linearized}. \footnote{The authors in \cite{brayton1964theory} consider $J_{2\star}=0_{n\times n}$.}
	
	\begin{theorem}\label{prop7}
		Let $\lambda\in\cmp$ be an eigenvalue of $\mathcal{A}$ and  $col(v,w)$ its corresponding eigenvector with $v,w\in\mathbb{C}^n$. Then, $\lambda$ lies on a circle centered in the point $(p_r,p_i)$ of the complex plane where
		\begin{equation}\label{pr}
				p_r:=\frac{v^*(\phi_M D_{d\star}\phi_M^\top)v}{\norm{v}^2},~
				p_i:=i\frac{v^*(\phi_M J_{2\star}\phi_M^\top)v}{\norm{v}^2}.
		\end{equation}
		Moreover, the radius of such circle is defined as
		\begin{equation}\label{radius}
			r_c:=\sqrt{\frac{\norm{(\tinv{\phi_M}\inv{M_\star}\phi_P^\top)w}^2-\norm{\Psi v}^2}{\norm{v}^2}+p_r^2+p_i^2}
		\end{equation}
		where $\Psi:=\phi_M(D_{d\star}-J_{2\star})\phi_M^\top$.
	\end{theorem}
	\begin{proof}
		Consider the eigenvalue problem $\mathcal{A}\begin{bmatrix}v\\w\end{bmatrix}=\lambda\begin{bmatrix}v\\w\end{bmatrix}$. Then, it follows that
		\begin{equation}\label{proof4}
			\begin{split}
				&\Psi v+ (\tinv{\phi_M}\inv{M}_\star\phi_P^\top) w=\lambda v\\
				\implies& \norm{(\Psi-\lambda I_n)v}^2 = \norm{(\tinv{\phi_M}\inv{M}_\star\phi_P^\top) w}^2\\
				\implies& |\lambda|^2-\frac{v^*(\Psi^\top+\Psi)v}{\norm{v}^2}\Re(\lambda)-\frac{iv^*(\Psi-\Psi^\top)v}{\norm{v}^2}\Im(\lambda)\\
				&=\frac{\norm{(\tinv{\phi_M}\inv{M}_\star\phi_P^\top) w}^2-\norm{\Psi v}^2}{\norm{v}^2}\\
				\implies&\Re (\lambda)^2+\Im (\lambda)^2-2\frac{v^*(\phi_M D_{d\star}\phi_M^\top)v}{\norm{v}^2}\Re(\lambda)\\
				&-2i\frac{v^*(\phi_M J_{2\star}\phi_M^\top)v}{\norm{v}^2}\Im(\lambda)\\
				&=\frac{\norm{(\tinv{\phi_M}\inv{M}_\star\phi_P^\top) w}^2-\norm{\Psi v}^2}{\norm{v}^2}.
			\end{split}
		\end{equation}
		Subsequently, by completing the squares in \eqref{proof4} with the expressions defined in \eqref{pr}, we get the circle centered in $(p_r, p_i)$ with radius $r_c$ as defined in \eqref{radius} in the complex plane.\footnote{Note that $\frac{iv^*(\Psi^\top-\Psi)v}{\norm{v}^2}\in \rea$.}
	\end{proof}
	
	\begin{remark}
		A similar analysis for the location of the eigevanlues can be performed by using the Gershgorin circle theorem or Gershgorin-like theorems (see \cite{horn2012matrix, johnson1989gersgorin}), where the radius of the circle (resp. the center of the circle) from these theorems is characterized by the sum of the non-diagonal elements (resp. by the diagonal element). Whereas Theorem \ref{prop7} characterizes the radius and the center of the circle by using the norms of the sub-blocks of $\mathcal{A}$. 
		\\
	\end{remark}
	Theorem \ref{prop7} claims that each eigenvector of $\mathcal{A}$ defines a circle where its corresponding eigenvalue lies on. {This theorem provides a quick and intuitive visual way of studying the effect of the parameters on the spectrum of $\mathcal{A}$. Moreover, it aids in designing the controller gains without computing the eigenvectors, which may be cumbersome (especially when the matrix $\mathcal{A}$ is large).} In Fig. \ref{ceig}, we provide a visual example\footnote{Since the term $\frac{v^\star(\phi_M D_{d\star}\phi_M^\top)v}{\norm{v}^2}$ is positive definite, every circle of $-\mathcal{A}$ lies in the left-half plane.} for particular choices of $J_2(q,p)$.
	\begin{figure}[t]
		\centering
		\includegraphics[width=\columnwidth]{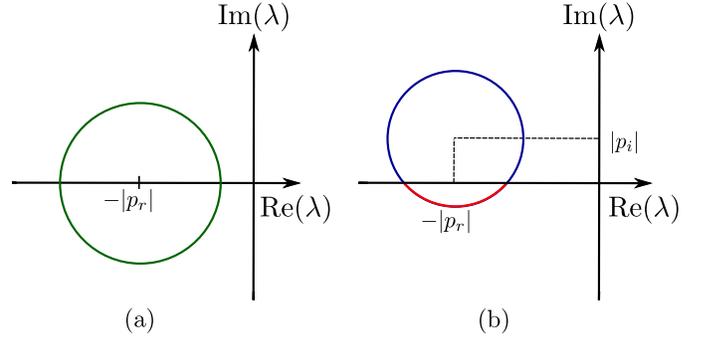}
		\caption{Circle containing an eigenvalue from of $-\mathcal{A}$. (a): ${J_{2\star}=0_{n\times n}}$. (b): $J_{2\star}\neq 0_{n\times n}$.}\label{ceig}
	\end{figure}	
	Note that the inclusion $J_2(q,p)$ leads to a more unpredictable oscillatory behavior of the closed-loop system as the circle \eqref{proof4} loses its symmetry on the real axis. Therefore, the oscillations may be increased in some coordinates and they can be reduced in others. For example, if the particular eigenvalue is located in the red arc, the oscillations may be reduced. Conversely, if that eigenvalue is located in the blue arc, the oscillations may be augmented.
	
	In general, the presence of $J_2(q,p)$ introduces a non-symmetric component to the block $(1,1)$ from $\mathcal{A}$. Therefore, in the sequel, we consider two scenarios for \eqref{linearized}: i) $\mathcal{A}$ with symmetric block $(1,1)$, and ii) $\mathcal{A}$ with non-symmetric block $(1,1)$.

	\subsection{A particular case: $\mathcal{A}$ with symmetric block $(1,1)$}\label{trtuning1}
	We recover \textit{a class of saddle point matrices} whose block (1,1) is symmetric for some particular choices of $J_2(q,p)$; for simplicity of exposition and without loss of generality, in the current section we consider $J_2(q,p)=0_{n\times n}$.\footnote{See \cite{benzi2006eigenvalues} for a complete description of the spectrum of $\mathcal{A}$ when ${J_2(q,p)=0_{n\times n}}$.} The spectrum can be characterized by the norm of each of its sub-block matrices; thus, it simplifies the analysis (especially when $\mathcal{A}$ is large). Moreover, the linearization of closed-loop system with PID-PBC (see \cite{borja2020}) recovers this particular class of saddle point matrices, which is later exploited in \cite{chan2020tuning} to provide tuning guidelines. Interestingly, the pioneering results of Brayton and Moser \cite{brayton1964theory} also use this form to study stability properties of electrical networks near their equilibrium point. 
	
	Here, we extend the methodologies described in \cite{chan2020tuning} for characterizing the spectrum to mechanical systems stabilizable by a broader class of PBC approaches that obtain the target dynamics as described in \eqref{tg1}-\eqref{tg2}.

	The following proposition provides a tuning rule that reduces the transient response oscillations to the minimum.
	
	\begin{theorem}\label{prop8}
		Consider the system \eqref{linearized}. If 
		\begin{equation}\label{p5}
			\arraycolsep=1pt \def\arraystretch{1.6}
			\begin{array}{rl}
				&\minlamb{D_{d\star}}^2\geq\\ 
				&4\maxlamb{ M_{d\star}\inv{M}_{\star}\nabla_{q}^2U_{d\star}\inv{M}_\star M_{d\star}}\maxlamb{M_{d\star}} 
			\end{array}
		\end{equation}
		holds, then the spectrum of $\mathcal{A}$ is real and positive. 
	\end{theorem}
	
	\begin{proof}
		Consider the eigenvalue problem $\mathcal{A}\begin{bmatrix}v\\w\end{bmatrix}=\lambda\begin{bmatrix}v\\w\end{bmatrix}$ with $\lambda\in\cmp$ and $v,w\in\cmp^{n\times n}$. Some computations yield the quadratic equation
		\begin{equation*}
			\def\arraystretch{1.6}
			\begin{array}{ll}
				\lambda^2-\frac{v^*(\phi_MD_{d\star}\phi_M^\top)v}{\norm{v}^2}\lambda
				+\frac{v^*(\tinv{\phi_M}\inv{M_\star}\nabla_q^2U_{d\star}\inv{M}\inv{\phi_M})v}{\norm{v}^2}=0;
			\end{array}
		\end{equation*} 
		whose solution is given by
		\begin{equation*}
			\arraycolsep=1pt \def\arraystretch{1.6}
			\begin{array}{ll}
				\lambda&=\frac{1}{2}\frac{v^*(\phi_MD_{d\star}\phi_M^\top)v}{\norm{v}^2}\pm\\
				&\frac{1}{2}\sqrt{\left(\frac{v^*(\phi_MD_{d\star}\phi_M^\top)v}{\norm{v}^2}\right)^2-4\frac{v^*(\tinv{\phi_M}\inv{M_\star}\nabla_q^2U_{d\star}\inv{M}\inv{\phi_M})v}{\norm{v}^2}}
			\end{array}.
		\end{equation*}
		Note that $\lambda$ is real if and only if the discriminant of the previous solution is non-negative, that is
		\begin{equation}\label{sol}
			\left(\frac{v^*(\phi_MD_{d\star}\phi_M^\top)v}{\norm{v}^2}\right)^2\geq 4\frac{v^*(\tinv{\phi_M}\inv{M_\star}\nabla_q^2U_{d\star}\inv{M}\inv{\phi_M})v}{\norm{v}^2}
		\end{equation}
		holds. In order to verify \eqref{sol} for the entire spectrum of $\mathcal{A}$, consider the change of coordinates $v_1= \phi_M^\top v$, then it follows that
		\begin{equation*}
				\frac{(v^*_1 D_{d\star}v_1)^2}{(v^*_1 M_{d\star} v_1)^2}\geq 4 \frac{(v^*_1 M_{d\star}\inv{M}_{\star}\nabla_{q}^2U_{d\star}\inv{M}_\star M_{d\star}v_1)}{(v^*_1 M_{d\star} v_1)},
		\end{equation*}
		and by multiplying both sides with $\frac{1}{\norm{v_1}^4}$, it follows that
		\begin{equation}\label{discri}
			\frac{(v^*_1 D_{d\star}v_1)^2}{\norm{v_1}^4}\geq 4 \frac{(v^*_1 M_{d\star}\inv{M}_{\star}\nabla_{q}^2U_{d\star}\inv{M}_\star M_{d\star}v_1)(v^*_1 M_{d\star} v_1)}{\norm{v_1}^4}.
		\end{equation}
		Then, note that we have the following inequalities
		\begin{equation*}
			\frac{(v^*_1 D_{d\star}v_1)^2}{\norm{v_1}^4}\geq \minlamb{D_{d\star}}^2
		\end{equation*}
		and
		\begin{equation*}
			\frac{(v^*_1 M_{d\star}\inv{M}_{\star}\nabla_{q}^2U_{d\star}\inv{M}_\star M_{d\star}v_1)(v^*_1 	M_{d\star} v_1)}{\norm{v_1}^4}\\
		\end{equation*}
		\begin{align*}
			&\leq \maxlamb{ M_{d\star}\inv{M}_{\star}\nabla_{q}^2U_{d\star}\inv{M}_\star M_{d\star}}\maxlamb{M_{d\star}} 
		\end{align*}
		Therefore, it follows that if $\eqref{p5}$ holds; then, \eqref{discri} holds for every $\lambda$. Additionally, note that $\lambda\in\rea_+$ since $D_d\succ0$; hence, the spectrum of $\mathcal{A}$ is positive.
	\end{proof}
	\begin{remark}\label{r10}
		Another convenient inequality for implementation purposes stemming from \eqref{discri} is given by
		\begin{equation*}
			\minlamb{D_{d\star}}^2\geq  4\frac{\maxlamb{M_{d\star}}^3\maxlamb{\nabla_{q}^2U_{d\star}}}{\minlamb{M_\star}^2}.
		\end{equation*}
		However, the obtained values may be conservative.
		\\
	\end{remark}
	
	By employing Theorem \ref{prop8} as a tuning rule, we can ensure that the spectrum of $\mathcal{A}$ is real and positive; hence, the linearized system \eqref{linearized} does not exhibit oscillations in its transient response.
	 
	Note that there is another set of assignable variables, that is, $\maxlamb{D_{d\star}},\minlamb{M_{d\star}}$, and $\minlamb{\nabla_{q}^2 U_{d\star}}$. In the following proposition, we exploit these variables for assigning a ``worst'' (or upper bound) value for the rise time of systems that verify \eqref{p5}. To this end, we define the rise time as the time taken by the system to reach $98\%$ of its desired value.
	\begin{theorem}\label{thm_risetime}
	 	{Assume that \eqref{p5} holds. Denote with ${\lambda_{tr}\in\rea_+}$ the lower bound of the spectrum of $\mathcal{A}$ given by
	 	\begin{equation}\label{ltr}
	 			\lambda_{tr}:=\begin{cases}
	 				\min\left\{\minlamb{\phi_MD_{d\star}\phi_M^\top},\frac{2\delta}{1+\sqrt{\Delta}}\right\},~\text{if}~\Delta\geq0\\
	 				\minlamb{\phi_MD_{d\star}\phi_M^\top},~\text{otherwise}
	 			\end{cases}	
	 		\end{equation}
	 		where $$\delta:=\minlamb{\phi_P\inv{M}_\star M_{d\star}\inv{D_{d\star}}M_{d\star}\inv{M}_\star\phi_P^\top},$$
	 		and
	 		$$\Delta:=1-4\dfrac{\delta}{\maxlamb{\phi_MD_{d\star}\phi_M^\top}}.$$
	 	Then, the rise time of \eqref{linearized} with $J_{2\star}=0_{n\times n}$ is bounded by }
	 	\begin{equation}\label{rise}
	 		\tilde{t}_{rt}:=\frac{4}{\lambda_{tr}}.
	 	\end{equation}
	 \end{theorem}
     ~\\
	 \begin{proof}
	 	Expression \eqref{ltr} follows from direct implementation of Theorem 2.1 from \cite{shen2010eigenvalue}. Subsequently, note that the trajectories of the closed-loop reaching 98$\%$ of its steady-state value is equivalent to the decay function~~--~~that bounds the trajectories~~--~~at 2$\%$ of its initial value, i.e.,
	 	$$2\%\approx\exp\{-4\}.$$
	 	Thus, \eqref{rise} follows from 
	 	$$\exp\left\{-\lambda_{tr}\tilde{t}_{tr}\right\}=\exp\{-4\}.$$
	 \end{proof}

	\subsection{The general case:  $\mathcal{A}$ with non-symmetric block $(1,1)$} 
	In general, the addition of $J_2(q,p)$ introduces a non-symmetric component to the block $(1,1)$ of $\mathcal{A}$, which hinders the oscillatory behavior analysis of the transient response. Nevertheless, its addition may potentially have an interesting positive effect on the performance in terms of stabilization and oscillations. 
		
	Furthermore, we get interesting results if the gyroscopic terms are \textit{not intrinsic}~--~i.e.,
	$J_g(q,p)=0_{n\times n}$, see Proposition \ref{prop3}~--~as there is a transformation such that the new Hamiltonian is of the form of the kinetic plus the potential energy. Therefore, we can find a transformation of its linearized system such that it has a saddle point form whose block $(1,1)$ is symmetric. We summarize this result in the following proposition.

\begin{theorem}\label{prop10}
	There exists a transformation such that the linearization of \eqref{tg1}-\eqref{tg2} around the equilibrium $(q_\star,0_n)$ has a saddle point form whose block $(1,1)$ is symmetric if and only if the gyroscopic terms are \textit{non}-intrinsic.
\end{theorem}
\begin{proof}
	Consider the closed-loop system in the canonical Hamiltonian system form \eqref{syscanonical}-\eqref{syscanonical2} and define $J_2(q,p)$ as in \eqref{prop2eq}. Then, it follows that the linearization of the canonical Hamiltonian system around the equilibrium $(q_\star,0_n)$ corresponds to (recall that $\tilde{q}:=q-q_\star$ and $\tilde{p}:=p-p_\star$ with $p_\star=0_n$)
	\begin{equation*}
		\arraycolsep=2.6pt \def\arraystretch{1.6}
		\begin{array}{rcl}
			\begin{bmatrix}
				\dot{\tilde{q}}\\\dot{\tilde{p}}
			\end{bmatrix}&=&\begin{bmatrix}
				0_{n\times n}&I_n\\-I_n&-{D}_{c\star}
			\end{bmatrix}\nabla^2 H_{c\star}	\begin{bmatrix}
				\tilde{q}\\\tilde{p}
			\end{bmatrix},\\
		\nabla^2 H_{c\star}&=&\begin{bmatrix}
			\nabla_q U_{d\star}+B_{12}M_{c\star}B^\top_{12}&B_{12}\\B_{12}^\top&\inv{M}_{c\star}
		\end{bmatrix}
		\end{array}
	\end{equation*}
	with $B_{12}:=-(\nabla_q Q_{d})_\star^\top\inv{M}_{c\star}$
	Subsequently, let
	\begin{equation*}
		\inv{M}_{c\star}=\phi_{Mc}^\top\phi_{Mc}, \nabla_q^2U_{d\star}=\phi_{Pc}^\top\phi_{Pc}
	\end{equation*}
	where $\phi_{Mc},\phi_{Pc}\in\rea^{n\times n}$ are upper triangular matrices obtained from the Cholesky decomposition; and consider the similarity transformation matrix $T_c\in\rea^{2n\times 2n}$ and new coordinates $z_c\in\rea^{2n}$
	\begin{equation*}
		T_c:=\begin{bmatrix}
			\phi_{Pc}&0_{n\times n}\\-\phi_{Mc}(\nabla_qQ_d)_\star&\phi_{Mc}
		\end{bmatrix},~z_c=T_c\begin{bmatrix}
			\tilde{q}\\\tilde{p}
		\end{bmatrix}.
	\end{equation*}
	The linearized system in the coordinates $z_c$ becomes
	\begin{equation*}
		\begin{array}{rl}
			\dot{z}_c&=-\mathcal{A}_c z_c,\\
			\mathcal{A}_c&:=\\
			&\begin{bmatrix}
				\phi_{Mc}\left[{D}_{c\star} + (\nabla_q Q_d)_\star-(\nabla_q Q_d)_\star^\top\right]\phi_{Mc}^\top&\phi_{Mc}\phi_{Pc}^\top\\
				-\phi_{Pc}\phi_{Mc}^\top&0_{n\times n}
			\end{bmatrix}
		\end{array}
	\end{equation*}
	Thus, $\mathcal{A}_c$ is a saddle point matrix with symmetric block $(1,1)$ if and only if the gyroscopic terms are \textit{non-intrinsic}, i.e.,
	
	\begin{equation}\label{nablaq2}
		\nabla_q Q_d(q)=(\nabla_q Q_d(q))^\top.
	\end{equation}
\end{proof}

Therefore, Theorem \ref{prop10} suggests that when $J_2(q,p)$ is chosen as in \eqref{prop2eq} with \textit{non-intrinsic} gyroscopic terms (i.e., $Q_d(q)$ satisfying \eqref{nablaq2}); then, we can implement the tuning rules as described in Section \ref{trtuning1}.

	\begin{remark}
		We can also prescribe the oscillation behavior to the transient response of the linearized system \eqref{linearized} by employing the well-known eigenvalue assignment (or pole placement) methodology, which is also known as the inverse problem for damped (gyroscopic) systems, i.e., given the complete set of eigenvalues and eigenvectors, find  $M_{d\star}$, $J_{d\star}$, $D_{d\star}$, and $\nabla_{q}^2 U_{d\star}$. However, we underscore that this tuning methodology lacks physical intuition.
\end{remark}
{
\begin{remark}\label{damping_rem}
The tuning rules described in this section require a proper characterization of $M(q)$, $U(q)$ and $D(q,p)$. Although the mass-inertia matrix and potential energy can be obtained relatively easily compared to $D(q,p)$, obtaining the natural damping can be challenging. Nonetheless, the tuning rules work even with a rough estimate as the closed-loop remains stable. However, we underscore that it may change the oscillatory behavior (the system may become overdamped or underdamped). The next section discusses some practical remarks on the damping treatment.
\end{remark} }

	\section{On damping treatment}\label{treatment}
	The implementation of the tuning rules described in Section \ref{trtuning} requires certain knowledge of the parameters of the open-loop system, i.e., the stiffness, mass inertia, and damping matrices. Albeit characterizing the former two parameters remain relatively easy, identifying the damping phenomenon is still an open question due to its complex nonlinear behavior.  Hence, to ensure the accuracy of our tuning methodology, we describe some methodologies found in the literature to identify the damping matrix in the current section.
	
	Some detailed assessments for general damping identification methods can be found in \cite{adhikari2001identification, prandina2009assessment}. Most of these damping identification techniques are based on the well-established modal analysis tool.  The tool characterizes the dynamics of a physical structure~--~with the help of the acquired data~--~in terms of the modal parameters such as natural frequency, damping factor, eigenvalues, and eigenvectors (or mode shape). The main disadvantage of the modal analysis tool is that it requires different equipment types, e.g., tens of sensors, impact hammer, and data acquisition hardware.  On the other hand, in \cite{liang2007damping} we find a simple damping identification method based on the energy of the system that simplifies the data recollection process. Such a methodology requires only one set of measurements~--~position, velocity, and acceleration data~--~but a larger set improves the identification accuracy. However, this damping identification methodology is restricted to constant and diagonal mass inertia and stiffness matrices. An extension of such results to a larger class of mechanical systems is the energy-based damping identification (EBDI) approach which can be found in \cite{chan2021passivity}.
	
	Nevertheless, the mentioned references only characterize linear damping (or viscous damping)  matrices. Therefore, these identifications are only valid in a region near the equilibrium point. Thus, the accuracy of the tuning rules may decrease when the trajectories start far from the equilibrium point. Also, note that the damping inaccuracy characterization can be included as part of $\hat{d}(t,x)$ in \eqref{transformed2}. Then it follows that due to this disturbance, the closed-loop system may not converge to the equilibrium point. To overcome this issue, the author in \cite{chen2004disturbance} proposes a nonlinear disturbance observer-based control (NDOBC). This dynamic extension approach is twofold: i) estimating the disturbance and ii) compensating the estimated disturbance using proper feedback. In \cite{sandoval2011interconnection, fu2018nonlinear}, the authors extend such an approach to port-Hamiltonian systems.
	
	Thus, we can ensure the accuracy of the tuning rules described in Section \ref{trtuning} by selecting a proper damping identification methodology (the selection process may depend on the availability of the equipment). Moreover, if the damping matrix is highly nonlinear, we can implement the chosen identification methodology combined with the NDOBC approach.
	
	\section{Case studies}\label{case}
	{In this section, we illustrate the applicability of some of our tuning rules for fully-actuated and underactuated mechanical systems. For the former, we employ the results from Sections  \ref{nonperturbed} and \ref{trtuning1}.} For the latter,  we employ the tuning rules discussed in Sections \ref{perturbed} and \ref{trtuning}. For both configurations, we employ a PBC approach with the form, \footnote{This controller structure can be found in PBC approaches such as in \cite{wesselink} or \cite{gomez2004physical} we refer the reader to these results for further details.}
	\begin{equation}\label{control}
		u=-K_{\tt{es}}G^\top( q-q_\star)-K_{\tt{di}}G^\top \dot{q}-K_{\tt{int}}G^\perp \dot{q}+\tilde{\kappa}(q)
	\end{equation}
	where $K_{\tt{es}}, K_{\tt{di}}$ are positive definite matrices with appropriate dimensions; $K_{\tt{int}}$ is a matrix with appropriate dimensions that modifies the interconnection of the system; and $\kappa:\rea^n\to \rea^m$ is a vector to be defined later. Moreover, we employ robot manipulators that verify \eqref{boundM}.  
	
	\subsection{Fully-actuated Mechanical System: A 5-DoF Robotic Arm}
	To demonstrate the effectiveness of the rules proposed in Sections \ref{nonperturbed} and \ref{trtuning1}, we use the Philips Experimental Robotic Arm (PERA) as shown in Fig. \ref{pera} (see \cite{rijs2010philips}). 
	\begin{figure}[t]
		\centering
		\includegraphics[width=0.3\textwidth]{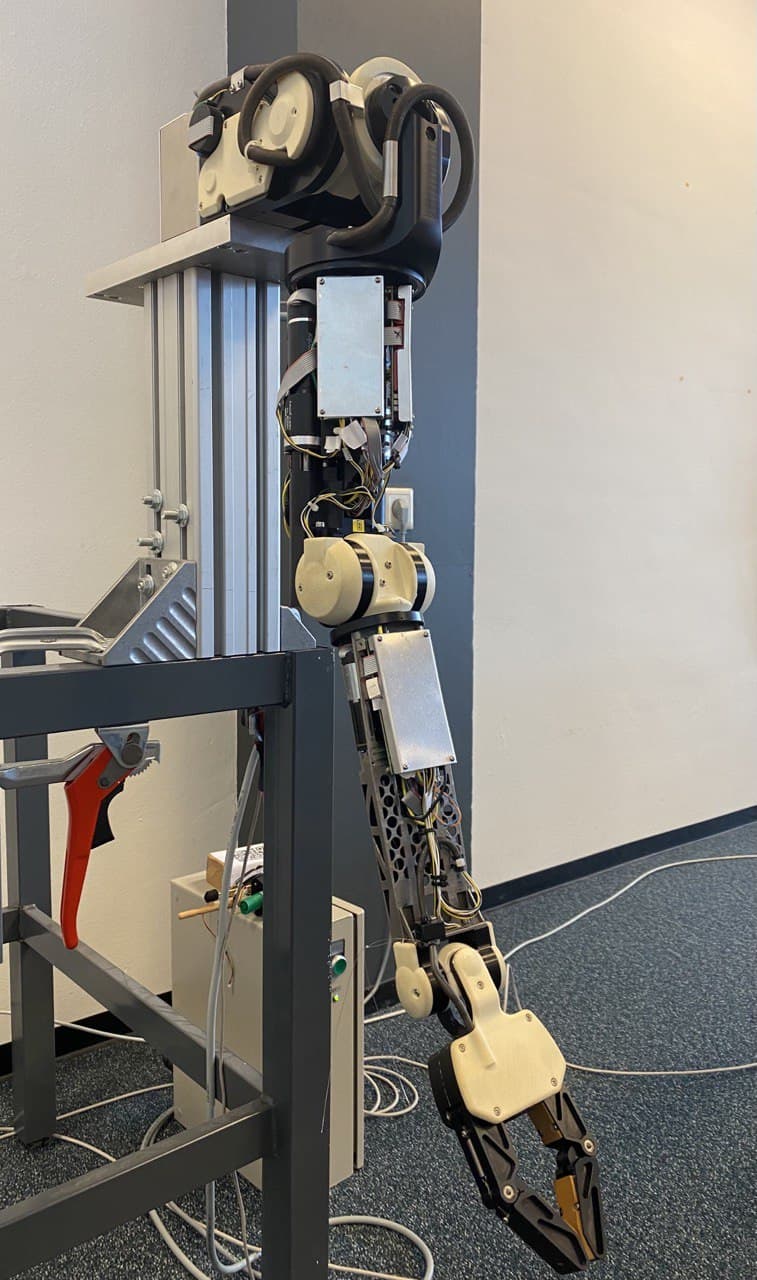}
		\caption{Experimental setup: Philips Experimental Robotic Arm (PERA)}\label{pera}
	\end{figure}
	For our experiments, we select the 5-DoF configuration with the following joints
	\begin{itemize}
		\item Shoulder yaw with angle $q_1$.
		\item Shoulder pitch with angle $q_2$.
		\item Shoulder roll with angle $q_3$.
		\item Elbow pitch with angle $q_4$.
		\item Elbow roll with angle $q_5$.
	\end{itemize} 
	The model is given by \eqref{sysmec} with $n=m=5$, $D(q,p)=0_{5\times 5}$, and $G=I_5$. Furthermore, the expressions for $M(q)$ and $U(q)$ are omitted due to space constraint; we refer the reader to \cite{ariasenergy,bol2012force} for further details. Additionally, a MATLAB\textsuperscript{\textregistered} script to generate the latter expressions can be found in \cite{bol2}. 
	We stabilize the PERA at the desired configuration ${q_\star=col(0.5,0.6,-1.6,1.3,0.5)~rad}$ with the controller \eqref{control} with $K_{\tt{int}}=0_{5\times 5}$ and $\tilde{\kappa}(q)=\nabla_q U(q)$. Thus, the targets dynamics are described as in \eqref{tg1}-\eqref{tg2} with  
	\begin{equation}\label{idapbcgains}
		\arraycolsep=1pt \def\arraystretch{1.6}
		\begin{array}{rcl}
			J_2(q,p)&=&0_{5\times 5}\\
			D_d(q,p)&=&D(q,p)+GK_{\tt{di}}G^\top\\
			M_d(q)&=&M(q),\\
			U_d(q)&=&\frac{1}{2}(q-q_\star)^\top GK_{\tt{es}}G^\top(q-q_\star). \\   
		\end{array}
	\end{equation}
	{We consider a rough estimate of the damping for this experimental setup. As mentioned in Remark \ref{damping_rem}, the tuning rules require a proper characterization of the parameters of the mechanical system and obtaining them in practice may be challenging. However, we demonstrate that the guidelines work even with an approximation. The gains selection for this case study are shown in Table \ref{gains}. \footnote{For simplicity, we have selected $K_{\tt{es}}$ and $K_{\tt{di}}$ as diagonal matrices.} 	
	\begin{table}[t]
		\def\arraystretch{1.6}
		\caption{PERA controller gains}\label{gains}
		\centering
		\begin{tabular}{lll}
			\hline
			& \multicolumn{1}{c}{$K_{\tt di}$} & \multicolumn{1}{c}{$K_{\tt es}$}     \\ \hline
			Case A & diag\{1,1,1,1,1\}       & diag\{200,175,200,175,200\} \\
			Case B & diag\{20,30,30,30,30\}  & diag\{200,175,200,175,200\} \\
			Case C & diag\{20,30,30,30,30\}  & diag\{100,100,100,100,100\} \\ \hline
		\end{tabular}
	\end{table}	
		
	First, the response for Case A corresponds to an arbitrary tuning for comparison purposes, where we can see the oscillations (and overshoot) for every joint in Fig.~\ref{exp_data}.	
	\begin{figure}[t]
		\centering
		\includegraphics[width=0.9\columnwidth]{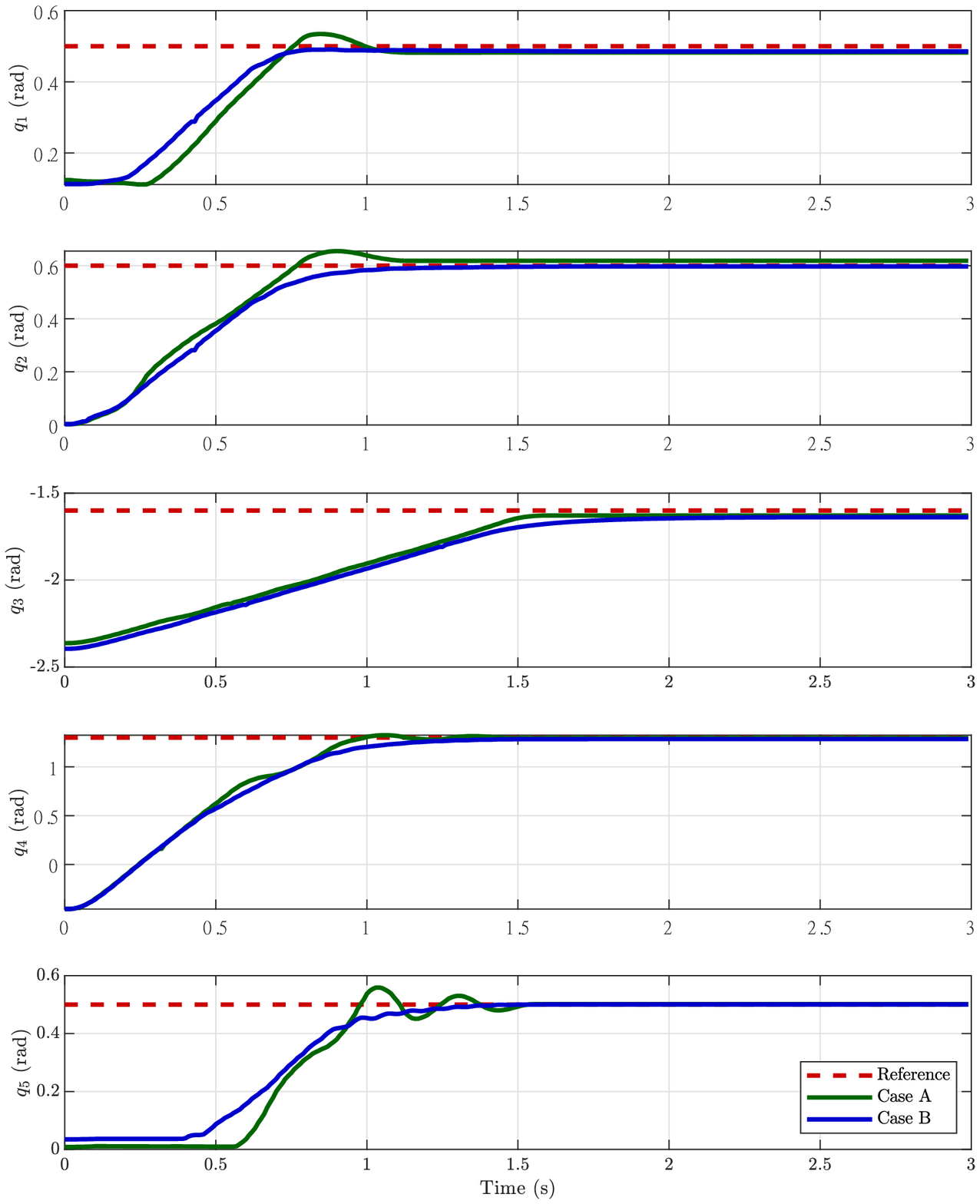}
		\caption{Trajectories for angular position of the shoulder joints for Case A and B.}\label{exp_data}
	\end{figure}	
	Then, to remove the oscillations from the transient response, we employ \eqref{p5} from Theorem \ref{prop8} to select the parameters of the control law \eqref{control}. Note that from \eqref{idapbcgains}, we get that  $M_{d\star}=M_\star$, $D_{d\star}=K_{\tt{di}}$, and} {${\nabla_q^2U_{d\star}=K_{\tt{es}}}$. Subsequently, to calculate $\minlamb{K_{\tt{di}}}$, we} {fix $K_{\tt{es}}=diag\{200,175,200,175,200\}$.\footnote{We have chosen this $K_{\tt{es}}$ value since we obtain an acceptable response for any $\minlamb{K_{\tt{di}}}$.} Therefore, we get that 
	\begin{equation*}
		\minlamb{K_{\tt{di}}}=19.8731.
	\end{equation*}
	Next, for Case B, we have selected $K_{\tt di}$ according to the calculated above. Note that, in Fig.~\ref{exp_data}, the oscillations (and overshoot) from Case B are attenuated in comparison with Case A. Moreover, in Case B, the responses are slightly overdamped,} {suggesting that the calculated $K_{\tt{di}}$ is conservative and the natural damping is different than zero, i.e., $D(q)\succ 0$ (see Remark \ref{damping_rem}). Nonetheless, even with a rough estimate, this tuning approach ensures that the oscillations in the closed loop are removed.

	Then, to illustrate the applicability of \eqref{rate} as a tuning rule, we use the Case B as the new reference baseline and we decrease the term $\beta_{\max}$ for Case C (recall from \eqref{betas} that $\beta_{\max}:=\nabla_{q}^2U_{d\star}$). Since $\beta_{\max}$ is proportional to the upper bound of the rate of convergence, it is expected that the rate of convergence reduces in Case C with respect to Case B, which is verified in Fig.~\ref{exp_data_pera2}.
	\begin{figure}[t]
		\centering
		\includegraphics[width=0.9\columnwidth]{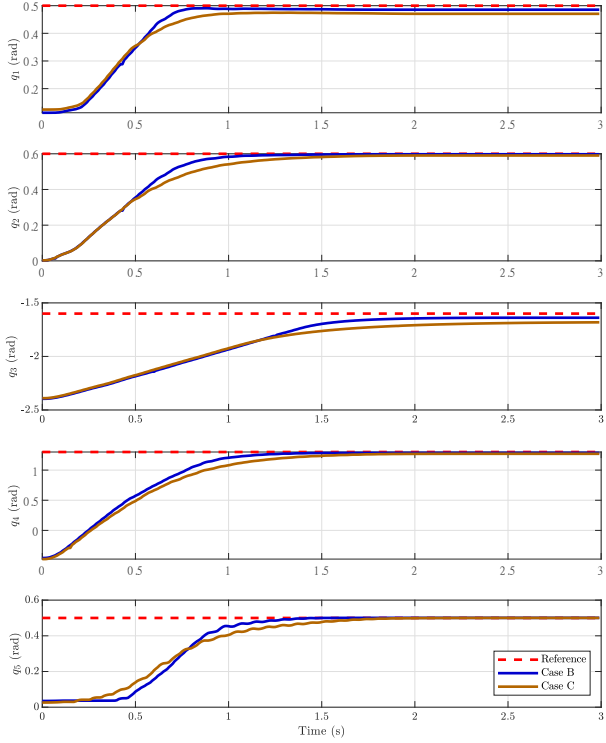}
		\caption{Trajectories for angular position of the shoulder joints for Case B and C.}\label{exp_data_pera2}
	\end{figure}	

	There is a small steady-state error in the joint positions, which may be due to non-modeled physical phenomena such as measurement noise, dry friction, or asymmetry of the motors. We underscore that the implemented PBC scheme is applied to the passive output signal, which corresponds to the actuated velocities for mechanical systems. Therefore, the integral action is applied on the actuated velocities and this scheme does not compensate for the error seen in the trajectories of the positions. For further details on compensating the position error, we refer the reader to \cite{dirksz2012power,romero2013robust,chan2022integral}.}

\begin{remark}
	We omit the use of Theorem \ref{thm_risetime}, i.e., the rise time tuning rule, as the results for this particular study case are highly conservative. However, it may be relevant for another set of mechanical system parameters, see the example in \cite{chan2020tuning}.
\end{remark}

\subsection{Underactuated Mechanical System: A 2-DoF Planar Manipulator with flexible joints}

\color{black}
This subsection illustrates the guidelines described in Section \ref{perturbed} and \ref{trtuning}. To this end, we employ a 2 DoF planar manipulator with flexible joints as shown in Fig. \ref{man} (see \cite{quanser} for the reference manual).
The manipulator in closed-loop with \eqref{control} is described in \eqref{disturbed} with a disturbance given by $d=col(0,0,0.5,0.5)~Nm$, \footnote{The open-loop model is described as in \eqref{sysmec} with $n=4,m=2$, which corresponds to an underactuated configuration since $m<n$.} and we have that
\begin{itemize}
\item $q_1$: Position of the link 1.
\item $q_2$: Position of the link 2.
\item $q_3$: Position of the motor of link 1.
\item $q_4$: Position of the motor of link 2.
\end{itemize}

Note that $col(q_1,q_2)$ corresponds to the unactuated coordinates. The rest of the parameters~--~obtained from \cite{chan2021passivity}~--~are given as
\begin{equation*}
	\def\arraystretch{1.6}
	\begin{array}{rl}
		G&=\begin{bmatrix}
			0_{2\times2}\\ G_1
		\end{bmatrix},~G_1=diag\{1,1.67\},\\ 
		U(q)=&\frac{1}{2}\norm{col(q_1,q_2)-col(q_3,q_4)}_{K_s}^2,\\
		K_s&=diag\{8.43, 16.86\},\\
		D&= diag\{D_u,D_a\},\\
		D_u&=diag\{0.0331, 0.0077\},\\
		D_a&=diag\{2.9758, 2.8064\},\\
		M(q)&=\begin{bmatrix}M_l(q_2)& 0_{2\times 2}\\0_{2\times 2}& M_m\end{bmatrix},\\
		M_m&=diag\{0.0628, 0.0026\},\\
		M_l(q_2)&:=\begin{bmatrix}
			a_1+a_2+2b\cos(q_2)&a_2+b\cos(q_2)\\
			a_2+b\cos(q_2)&a_2
		\end{bmatrix},		
	\end{array}
\end{equation*} 
where $a_1=0.1547,~a_2=0.0111$, $b=0.0168$.

\begin{figure}[t]
	\centering
	\includegraphics[width=0.9\columnwidth]{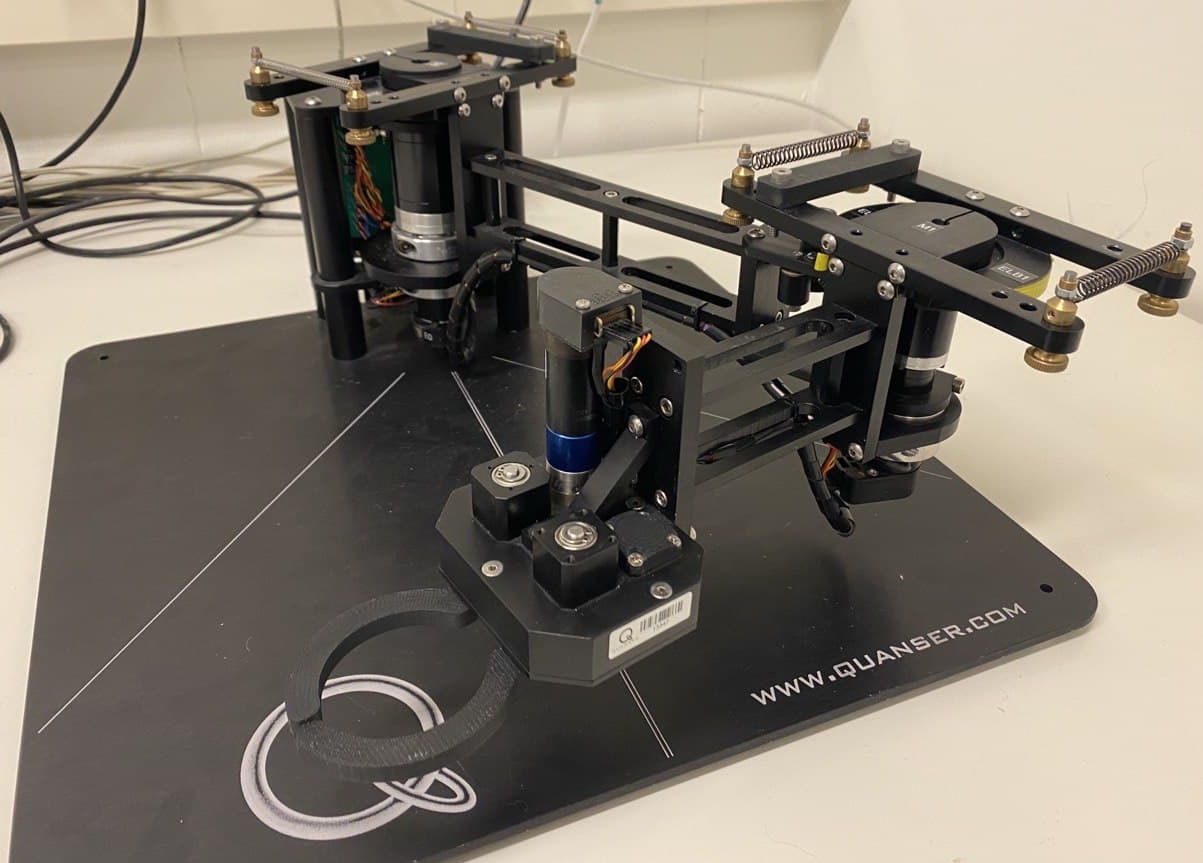}
	\caption{2-DoF Planar Manipulator}\label{man}
\end{figure}	
The manipulator is stabilized at the desired configuration $q_\star=col(0.6,0.8,0.6,0.8)~rad$ by employing the controller \eqref{control} with $\tilde{\kappa}(q)=0_2$.  The corresponding closed loop are described as in \eqref{tg1}-\eqref{tg2} with 
\begin{equation*}
	\arraycolsep=10pt \def\arraystretch{1.6}
	\begin{array}{rcl}
		J_2&=&\begin{bmatrix}
			0_{2\times 2} &\frac{1}{2}(G_1K_{\tt{int}})^\top\\
			-\frac{1}{2}G_1K_{\tt{int}}&0_{2\times 2}
		\end{bmatrix}\\
		D_d&=&\begin{bmatrix}
			D_u &\frac{1}{2}(G_1K_{\tt{int}})^\top\\
			\frac{1}{2}G_1K_{\tt{int}}&D_a+G_1K_{\tt{di}}G_1^\top
		\end{bmatrix}\\
		M_d(q)&=&M(q),\\
		U_d(q)&=&\frac{1}{2}(q-q_\star)^\top GK_{\tt{es}}G^\top(q-q_\star)+U(q). 
	\end{array}
\end{equation*}
\begin{table}[t]
	\centering
	\caption{Planar manipulator controller gains}\label{mangain}
	\begin{tabular}{cccc}
		\hline
		& $K_{\tt di}$             &  $K_{\tt es}$           & $K_{\tt int}$             \\ \hline
		Case D & diag\{1.5,1.5\} & diag\{3.5\}   & diag\{0,0\}      \\
		Case E & diag\{1.5,1.5\} & diag\{12,15\} & diag\{0,0\}      \\
		Case F & diag\{7,5\}     & diag\{12,15\} & diag\{0,0\}      \\
		Case G & diag\{7,5\}     & diag\{12,15\} & diag\{1.1,0.43\} \\ \hline
	\end{tabular}
\end{table}
{Now, we proceed to improve the performance of the closed loop. First, we tune the nonlinear stability margin \eqref{sm} and the maximum overshoot \eqref{ov} by modifying $\beta_{\max}$.}
\newpage
{Recall that $\beta_{\max}$ is given by $K_{\tt es},$ which shapes the potential energy of the manipulator, see \eqref{betas}. Note that $\beta_{\max}$ (or $K_{\tt es}$) is proportional to the gain margin and the maximum overshoot; therefore, the gain margin is expected to increase at the expense of a higher overshoot when $\beta_{\max}$ is augmented. To highlight the mentioned, we have selected two sets of gains as shown in Case D and Case E from Table \ref{mangain}. We choose the Case D as the baseline case, then, we increment $\beta_{\max}$ in Case E by augmenting $K_{\tt es}$. The responses in Fig.~\ref{quanser1} verify that we obtain a better attenuation of the disturbance~--~i.e., the steady state error improves~--~at the expense of an increased overshoot.
\begin{figure}[t]
	\centering
	\includegraphics[width=\columnwidth]{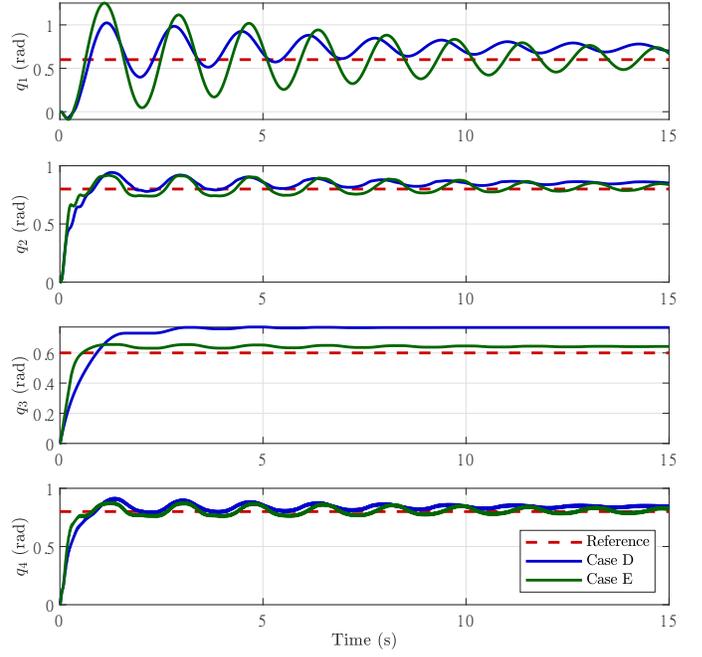}
	\caption{Trajectories for angular position for Case D and E.}\label{quanser1}
\end{figure} 
\begin{figure}[t]
	\centering
	\includegraphics[width=\columnwidth]{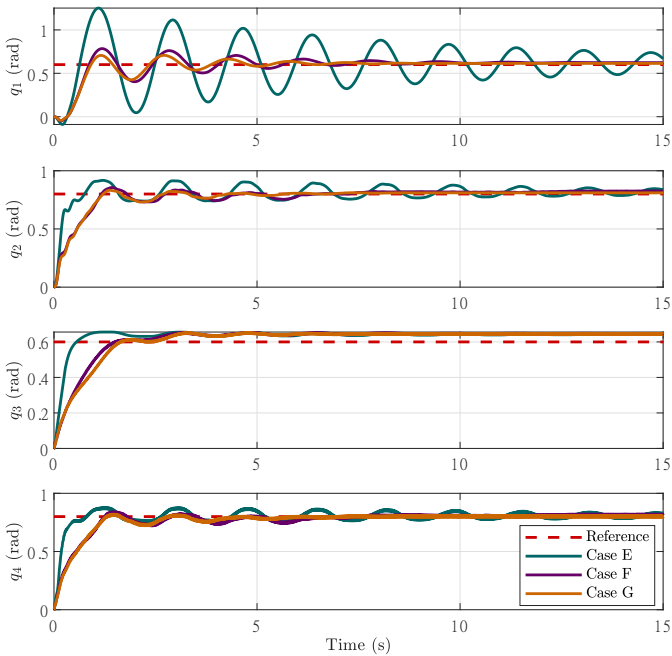}
	\caption{Trajectories for angular position for Case E, F, and G.}\label{quanser2}
\end{figure}

Next, we proceed to reduce the overshoot. Towards this end, we select the Case E as the new baseline and modify $K_{\tt di}$ and $K_{\tt int}$} {in  Case F and Case G.  The gains and responses for these cases are shown in Table \ref{mangain} and Fig.~\ref{quanser2}, respectively.
	
The gains selection for Case F is based on the circle-like theorem described in Theorem \ref{prop7}, where we use Fig.~\ref{ceig}.a as a }{ quick visual aid. Note that, by augmenting $K_{\tt di}$,  some of the circles shift to the left, leading to a faster convergence rate and, consequently, fewer oscillations. The latter is verified in Fig.~\ref{quanser2}, where the oscillations are reduced substantially in Case F with respect to Case E.  

Then, to improve further the performance, we inject gyroscopic forces via $J_2(q,p)$ in Case~G, which employ again Theorem \ref{prop7} to tune $K_{\tt{int}}$. Consider Fig. \ref{ceig}.b and note that the} {inclusion of $J_2(q,p)$ (or equivalently, $K_{\tt{int}}$) shifts some of the circles in the imaginary axis; therefore, the eigenvalue may be located anywhere in the red arc resulting in less damping ratio, and consequently, the response may exhibit fewer oscillations. The oscillations are reduced slightly in Case G with respect to Case F, as verified in Fig.~\ref{quanser2} and Table \ref{energy}, where it contains the values of the square of the $\mathcal{L}_2$-norm. \footnote{The square of the $\mathcal{L}_2$-norm corresponds to the energy contained in the such signal.}  The improvement in oscillations when $K_{\tt int}$ is introduced is slight for this particular case study since this gain is small. This behavior stems from a limitation with the controller, that is, $K_{\tt int}$ is selected such that $D_d\succ 0$. To verify this condition, we employ the Schur complement analysis, i.e., $K_{\tt int}$ is chosen such that $$D_a+G_1K_{\tt{di}}G_1^\top-\frac{1}{2}G_1K_{\tt{int}}\inv{D_u}\frac{1}{2}G_1K_{\tt{int}}^\top\succ0.$$
Therefore, the selection of $K_{\tt int}$ cannot be incremented further as it is restricted to the previous condition. Nonetheless, the improvement is still evident. }

\color{black}
 A video of the experimental results can be found in \url{https://youtu.be/yUGs44K77wE}.
\begin{table}[t]
	\centering
	\caption{Energy contained in each signal}\label{energy}
	\begin{tabular}{ccc}
		\hline
		& Case F & Case G \\ \hline
		$q_1$ & 2776   & 2666   \\
		$q_2$ & 4663   & 4614   \\
		$q_3$ & 2915   & 2878   \\
		$q_4$ & 4544   & 4509   \\ \hline
	\end{tabular}
\end{table}

\section{Concluding remarks and future work}\label{conclusion}
We have provided a broad guide for tuning the control parameters of a class of PBC methodologies that preserve the mechanical structure and prescribe the desired behavior in terms of several attributes: upper bound of the rate of convergence, maximum permissible overshoot, gain margin and oscillations in the transient response. Moreover, we have associated the PBC parameters with the physical quantities of the closed-loop system (energy or damping). Therefore, we have endowed the parameter selection process of PBC approaches with more intuition.
	
Furthermore, we have shown how a class of gyroscopic forces affects the behavior of the closed loop near the equilibrium.
	
Additionally, we have successfully implemented our tuning rules on two mechanical systems setups: i) the PERA system (fully-actuated configuration using 5 DoF), and ii) a planar manipulator with flexible joints (underactuated configuration). In both cases, we have reduced the oscillations of the transient response. 
	
Regarding future work, we aim to find tuning rules to prescribe the behavior in the vicinity of the equilibrium of the closed-loop system when intrinsic gyroscopic forces are introduced. {Moreover, we aim to find methodologies to calculate the parameters $\mu$ and $\epsilon$}.
\bibliographystyle{ieeetr}
\bibliography{ref} 

\begin{thebibliography}{10}

\bibitem{duindam2009modeling}
V.~Duindam, A.~Macchelli, S.~Stramigioli, and H.~Bruyninckx, {\em Modeling and
  control of complex physical systems: the {port-Hamiltonian approach}}.
\newblock Springer Science \& Business Media, 2009.

\bibitem{van2014port}
A.~J. van~der Schaft and D.~Jeltsema, ``Port-{Hamiltonian} systems theory: An
  introductory overview,'' {\em Foundations and Trends in Systems and Control},
  vol.~1, no.~2-3, pp.~173--378, 2014.

\bibitem{ortega2013passivity}
R.~Ortega, J.~A.~L. Perez, P.~J. Nicklasson, and H.~J. Sira-Ramirez, {\em
  Passivity-based control of Euler-Lagrange systems: mechanical, electrical and
  electromechanical applications}.
\newblock Springer Science \& Business Media, 2013.

\bibitem{vanderSchaft2017}
A.~J. van~der Schaft, {\em L2-Gain and Passivity Techniques in Nonlinear
  Control}.
\newblock Springer International Publishing, 2017.

\bibitem{gomez2004physical}
F.~G{\'o}mez-Estern and A.~J. van~der Schaft, ``Physical damping in {IDA-PBC}
  controlled underactuated mechanical systems,'' {\em European Journal of
  Control}, vol.~10, no.~5, pp.~451--468, 2004.

\bibitem{romero2018global}
J.~G. Romero, A.~Donaire, R.~Ortega, and P.~Borja, ``Global stabilisation of
  underactuated mechanical systems via {PID} passivity-based control,'' {\em
  Automatica}, vol.~96, pp.~178--185, 2018.

\bibitem{acosta2005interconnection}
J.~A. Acosta, R.~Ortega, A.~Astolfi, and A.~D. Mahindrakar, ``Interconnection
  and damping assignment passivity-based control of mechanical systems with
  underactuation degree one,'' {\em IEEE Transactions on Automatic Control},
  vol.~50, no.~12, pp.~1936--1955, 2005.

\bibitem{viola2007total}
G.~Viola, R.~Ortega, R.~Banavar, J.~{\'A}. Acosta, and A.~Astolfi, ``Total
  energy shaping control of mechanical systems: simplifying the matching
  equations via coordinate changes,'' {\em IEEE Transactions on Automatic
  Control}, vol.~52, no.~6, pp.~1093--1099, 2007.

\bibitem{hamada2020passivity}
K.~Hamada, P.~Borja, J.~M.~A. Scherpen, K.~Fujimoto, and I.~Maruta,
  ``{Passivity-Based Lag-Compensators With Input Saturation for Mechanical
  Port-{Hamiltonian} Systems Without Velocity Measurements},'' {\em IEEE
  Control Systems Letters}, vol.~5, no.~4, pp.~1285--1290, 2020.

\bibitem{wesselink}
T.~Wesselink, P.~Borja, and J.~M.~A. Scherpen, ``Saturated control without
  velocity measurements for planar robots with flexible joints,'' in {\em 2019
  IEEE 58th Conference on Decision and Control (CDC)}, pp.~7093--7098, 2019.

\bibitem{borja2020}
P.~{Borja}, R.~{Ortega}, and J.~M.~A. {Scherpen}, ``{New Results on
  Stabilization of port-{Hamiltonian} Systems via {PID} Passivity-based
  Control},'' {\em IEEE Transactions on Automatic Control}, pp.~1--1, 2020.

\bibitem{Donaire2012}
A.~Donaire and T.~Perez, ``Dynamic positioning of marine craft using a
  {port-Hamiltonian} framework,'' {\em Automatica}, vol.~48, pp.~851--856, May
  2012.

\bibitem{romero2014globally}
J.~G. Romero, R.~Ortega, and I.~Sarras, ``A globally exponentially stable
  tracking controller for mechanical systems using position feedback,'' {\em
  IEEE Transactions on Automatic Control}, vol.~60, no.~3, pp.~818--823, 2014.

\bibitem{wen2004experimental}
J.~T. Wen and B.~Potsaid, ``An experimental study of a high performance motion
  control system,'' in {\em Proceedings of the 2004 American Control
  Conference}, vol.~6, pp.~5158--5163, IEEE, 2004.

\bibitem{bechlioulis2010prescribed}
C.~P. Bechlioulis and G.~A. Rovithakis, ``Prescribed performance adaptive
  control for multi-input multi-output affine in the control nonlinear
  systems,'' {\em IEEE Transactions on Automatic Control}, vol.~55, no.~5,
  pp.~1220--1226, 2010.

\bibitem{aastrom2004revisiting}
K.~J. {\AA}str{\"o}m and T.~H{\"a}gglund, ``Revisiting the ziegler--nichols
  step response method for {PID} control,'' {\em Journal of process control},
  vol.~14, no.~6, pp.~635--650, 2004.

\bibitem{chen1990back}
F.-C. Chen, ``Back-propagation neural networks for nonlinear self-tuning
  adaptive control,'' {\em IEEE control systems Magazine}, vol.~10, no.~3,
  pp.~44--48, 1990.

\bibitem{rodriguez2021self}
O.~Rodr{\'\i}guez-Abreo, J.~Rodr{\'\i}guez-Res{\'e}ndiz, C.~Fuentes-Silva,
  R.~Hern{\'a}ndez-Alvarado, and M.~D. C. P.~T. Falc{\'o}n, ``Self-tuning
  neural network {PID} with dynamic response control,'' {\em IEEE Access},
  vol.~9, pp.~65206--65215, 2021.

\bibitem{ortega2002interconnection}
R.~Ortega, A.~J. van~der Schaft, B.~Maschke, and G.~Escobar, ``Interconnection
  and damping assignment passivity-based control of {port-controlled
  Hamiltonian} systems,'' {\em Automatica}, vol.~38, no.~4, pp.~585--596, 2002.

\bibitem{kotyczka2013local}
P.~Kotyczka, ``Local linear dynamics assignment in {IDA-PBC},'' {\em
  Automatica}, vol.~49, no.~4, pp.~1037--1044, 2013.

\bibitem{ferguson2019kinetic}
J.~Ferguson, A.~Donaire, and R.~H. Middleton, ``Kinetic-potential energy
  shaping for mechanical systems with applications to tracking,'' {\em IEEE
  Control Systems Letters}, vol.~3, no.~4, pp.~960--965, 2019.

\bibitem{chan2021exponential}
C.~Chan-Zheng, P.~Borja, N.~Monshizadeh, and J.~M.~A. Scherpen, ``Exponential
  stability and tuning for a class of mechanical systems,'' in {\em 2021
  European Control Conference (ECC)}, pp.~1875--1880, IEEE, 2021.

\bibitem{jeltsema2004tuning}
D.~Jeltsema and J.~M.~A. Scherpen, ``Tuning of passivity-preserving controllers
  for switched-mode power converters,'' {\em IEEE Transactions on Automatic
  Control}, vol.~49, no.~8, pp.~1333--1344, 2004.

\bibitem{dirksz2013tuning}
D.~A. Dirksz and J.~M.~A. Scherpen, ``Tuning of dynamic feedback control for
  nonlinear mechanical systems,'' in {\em 2013 European Control Conference
  (ECC)}, pp.~173--178, 2013.

\bibitem{chan2020tuning}
C.~Chan-Zheng, P.~Borja, and J.~M.~A. Scherpen, ``Tuning rules for a class of
  passivity-based controllers for mechanical systems,'' {\em IEEE Control
  Systems Letters}, vol.~5, no.~6, pp.~1892--1897, 2021.

\bibitem{blankenstein2002matching}
G.~Blankenstein, R.~Ortega, and A.~J. van~der Schaft, ``The matching conditions
  of controlled {Lagrangians and IDA-passivity based control},'' {\em
  International Journal of Control}, vol.~75, no.~9, pp.~645--665, 2002.

\bibitem{chang2002equivalence}
D.~E. Chang, A.~M. Bloch, N.~E. Leonard, J.~E. Marsden, and C.~A. Woolsey,
  ``The equivalence of controlled {Lagrangian} and controlled {Hamiltonian}
  systems,'' {\em ESAIM: Control, Optimisation and Calculus of Variations},
  vol.~8, pp.~393--422, 2002.

\bibitem{woolsey2004controlled}
C.~Woolsey, C.~K. Reddy, A.~M. Bloch, D.~E. Chang, N.~E. Leonard, and J.~E.
  Marsden, ``Controlled lagrangian systems with gyroscopic forcing and
  dissipation,'' {\em European Journal of Control}, vol.~10, no.~5,
  pp.~478--496, 2004.

\bibitem{borja2022role}
P.~Borja, C.~Della~Santina, and A.~Dabiri, ``On the role of coupled damping and
  gyroscopic forces in the stability and performance of mechanical systems,''
  {\em IEEE Control Systems Letters}, vol.~6, pp.~3433--3438, 2022.

\bibitem{ghorbel1993positive}
F.~Ghorbel, B.~Srinivasan, and M.~W. Spong, ``On the positive definiteness and
  uniform boundedness of the inertia matrix of robot manipulators,'' in {\em
  Proceedings of 32nd IEEE Conference on Decision and Control}, pp.~1103--1108,
  IEEE, 1993.

\bibitem{zhang2017pid}
M.~Zhang, P.~Borja, R.~Ortega, Z.~Liu, and H.~Su, ``{PID} passivity-based
  control of {port-Hamiltonian} systems,'' {\em IEEE Transactions on Automatic
  Control}, vol.~63, no.~4, pp.~1032--1044, 2017.

\bibitem{venkatraman2010speed}
A.~Venkatraman, R.~Ortega, I.~Sarras, and A.~J. van~der Schaft, ``Speed
  observation and position feedback stabilization of partially linearizable
  mechanical systems,'' {\em IEEE Transactions on Automatic Control}, vol.~55,
  no.~5, pp.~1059--1074, 2010.

\bibitem{fujimoto2003trajectory}
K.~Fujimoto, K.~Sakurama, and T.~Sugie, ``Trajectory tracking control of
  port-controlled {Hamiltonian} systems via generalized canonical
  transformations,'' {\em Automatica}, vol.~39, no.~12, pp.~2059--2069, 2003.

\bibitem{chan2021passivity}
C.~Chan-Zheng, P.~Borja, and J.~M.~A. Scherpen, ``Passivity-based control of
  mechanical systems with linear damping identification,'' {\em
  IFAC-PapersOnLine}, vol.~54, no.~19, pp.~255--260, 2021.
\newblock 7th IFAC Workshop on Lagrangian and Hamiltonian Methods for Nonlinear
  Control LHMNC 2021.

\bibitem{khalil2002nonlinear}
H.~Khalil, {\em Nonlinear systems}, vol.~3.
\newblock Prentice hall Upper Saddle River, NJ, 2002.

\bibitem{sontag1995characterizations}
E.~D. Sontag and Y.~Wang, ``On characterizations of the input-to-state
  stability property,'' {\em Systems \& Control Letters}, vol.~24, no.~5,
  pp.~351--359, 1995.

\bibitem{horn2012matrix}
R.~A. Horn and C.~R. Johnson, {\em Matrix analysis}.
\newblock Cambridge university press, 2012.

\bibitem{grune2002input}
L.~Grune, ``Input-to-state dynamical stability and its lyapunov function
  characterization,'' {\em IEEE Transactions on Automatic Control}, vol.~47,
  no.~9, pp.~1499--1504, 2002.

\bibitem{benzi2005numerical}
M.~Benzi, G.~H. Golub, and J.~Liesen, ``Numerical solution of saddle point
  problems,'' {\em Acta numerica}, vol.~14, pp.~1--137, 2005.

\bibitem{benzi2006eigenvalues}
M.~Benzi and V.~Simoncini, ``On the eigenvalues of a class of saddle point
  matrices,'' {\em Numerische Mathematik}, vol.~103, no.~2, pp.~173--196, 2006.

\bibitem{brayton1964theory}
R.~Brayton and J.~Moser, ``A theory of nonlinear networks. i,'' {\em Quarterly
  of Applied Mathematics}, vol.~22, no.~1, pp.~1--33, 1964.

\bibitem{johnson1989gersgorin}
C.~R. Johnson, ``A gersgorin-type lower bound for the smallest singular
  value,'' {\em Linear Algebra and its Applications}, vol.~112, pp.~1--7, 1989.

\bibitem{shen2010eigenvalue}
S.-Q. Shen, T.-Z. Huang, and J.~Yu, ``Eigenvalue estimates for preconditioned
  nonsymmetric saddle point matrices,'' {\em SIAM journal on matrix analysis
  and applications}, vol.~31, no.~5, pp.~2453--2476, 2010.

\bibitem{adhikari2001identification}
S.~Adhikari and J.~Woodhouse, ``Identification of damping: part 1, viscous
  damping,'' {\em Journal of Sound and vibration}, vol.~243, no.~1, pp.~43--61,
  2001.

\bibitem{prandina2009assessment}
M.~Prandina, J.~E. Mottershead, and E.~Bonisoli, ``An assessment of damping
  identification methods,'' {\em Journal of Sound and Vibration}, vol.~323,
  no.~3-5, pp.~662--676, 2009.

\bibitem{liang2007damping}
J.-W. Liang, ``Damping estimation via energy-dissipation method,'' {\em Journal
  of sound and Vibration}, vol.~307, no.~1-2, pp.~349--364, 2007.

\bibitem{chen2004disturbance}
W.-H. Chen, ``Disturbance observer based control for nonlinear systems,'' {\em
  IEEE/ASME transactions on mechatronics}, vol.~9, no.~4, pp.~706--710, 2004.

\bibitem{sandoval2011interconnection}
J.~Sandoval, R.~Kelly, and V.~Santib{\'a}{\~n}ez, ``Interconnection and damping
  assignment passivity-based control of a class of underactuated mechanical
  systems with dynamic friction,'' {\em International Journal of Robust and
  Nonlinear Control}, vol.~21, no.~7, pp.~738--751, 2011.

\bibitem{fu2018nonlinear}
B.~Fu, Q.~Wang, and W.~He, ``Nonlinear disturbance observer-based control for a
  class of {port-controlled Hamiltonian} disturbed systems,'' {\em IEEE
  Access}, vol.~6, pp.~50299--50305, 2018.

\bibitem{rijs2010philips}
R.~Rijs, R.~Beekmans, S.~Izmit, and D.~Bemelmans, ``Philips experimental robot
  arm: User instructor manual,'' {\em Koninklijke Philips Electronics NV,
  Eindhoven}, vol.~1, 2010.

\bibitem{ariasenergy}
M.~Muñoz-Arias, {\em Energy-based control design for mechanical systems}.
\newblock PhD thesis, University of Groningen, April 2015.

\bibitem{bol2012force}
M.~Bol, ``Force and position control of the {Philips} experimental robot arm in
  an energy based setting,'' Master's thesis, University of Groningen, 2012.

\bibitem{bol2}
M.~Bol and M.~Muñoz-Arias, {\em {Model Generator for Philips Experimental
  Robotics Arm}}.
\newblock Available at \url{https://github.com/cachanzheng/PERA}.

\bibitem{dirksz2012power}
D.~A. Dirksz and J.~M.~A. Scherpen, ``Power-based control: Canonical coordinate
  transformations, integral and adaptive control,'' {\em Automatica}, vol.~48,
  no.~6, pp.~1045--1056, 2012.

\bibitem{romero2013robust}
J.~G. Romero, A.~Donaire, and R.~Ortega, ``Robust energy shaping control of
  mechanical systems,'' {\em Systems \& Control Letters}, vol.~62, no.~9,
  pp.~770--780, 2013.

\bibitem{chan2022integral}
C.~Chan-Zheng, M.~Muñoz-Arias, and J.~M.~A. Scherpen, ``Tuning rules for
  passivity-based integral control for a class of mechanical systems,'' {\em
  IEEE Control Systems Letters}, vol.~7, pp.~37--42, 2023.

\bibitem{quanser}
Quanser, ``2 {DOF Serial Flexible Joint, Reference Manual},'' {\em Quanser
  Inc.}, Doc. No. 800, Rev 1. 2013.

\end{thebibliography}

\end{document}